\documentclass[smallcondensed]{svjour3}
\usepackage{pifont}
\usepackage[cmex10]{amsmath}
\usepackage{bm,physics}
\usepackage{graphicx}
\usepackage{tikz}
\usepackage{tikz-cd}
\usepackage{xfrac}
\usepackage{amsfonts,amssymb,amscd}
\usepackage{algorithm}
\usepackage{algpseudocode}
\usepackage[pdfencoding=auto,psdextra]{hyperref}
\usepackage{url}
\usepackage{bookmark}
\usepackage[noadjust]{cite}
\usepackage{stmaryrd}
\usepackage{verbatim,color,colortbl}
\usepackage{rotating,tabularx}
\usepackage{soul}
\usetikzlibrary{fit}

\newcommand{\Mod}[1]{\ (\textup{mod}\ #1)}
\newcommand{\etal}{\emph{et al.}}
\newcommand{\eg}{\emph{e.g.}}
\newcommand{\ie}{\emph{i.e.}}
\def\F{{\mathbb {F}}}
\def\N{{\mathbb {N}}}

\def\v{{\mathbf{v}}}

\def\bS{{\mathbf{S}}}
\def\bR{{\mathbf{R}}}
\def\u{{\mathbf{u}}}
\def\s{{\mathbf{s}}}
\def\m{{\mathbf{m}}}
\def\0{{\mathbf{0}}}
\def\1{{\mathbf{1}}}
\def\a{{\mathbf{a}}}
\def\b{{\mathbf{b}}}

\def\mW{{\mathcal{W}}}
\def\mM{{\mathcal{M}}}
\def\mG{{\mathcal{G}}}
\def\mP{{\mathcal{P}}}

\newcommand{\bbra}[1]{\left\llbracket{#1}\right\rrbracket}
\newcommand{\floor}[1]{\left\lfloor{#1}\right\rfloor}

\newcommand{\In}{\mathop{{\tt Inv}}}
\newcommand{\flip}{\mathop{{\tt Flip}}}
\newcommand{\double}{\mathop{{\tt Double}}}
\newcommand{\lra}[1]{\underleftrightarrow{#1}}
\journalname{}

\begin{document}
\title{Binary de Bruijn Sequences via Zech's Logarithms}
\author{Zuling Chang \and Martianus Frederic Ezerman \and Adamas Aqsa Fahreza \and San Ling 
\and\\Janusz Szmidt \and Huaxiong Wang}
\authorrunning{Z.~Chang et al.}
\institute{
Z. Chang \at School of Mathematics and Statistics, Zhengzhou University, Zhengzhou 450001, China\\
\email{zuling\textunderscore chang@zzu.edu.cn}
\and
M. F. Ezerman \and A. A. Fahreza \and S. Ling \and H. Wang \at Division of Mathematical Sciences, School of Physical and Mathematical Sciences,\\
Nanyang Technological University, 21 Nanyang Link, Singapore 637371\\
\email{\{fredezerman,adamas,lingsan,HXWang\}@ntu.edu.sg}
\and
J.~Szmidt \at Military Communication Institute, ul. Warszawska 22 A, 05-130 Zegrze, Poland\\
\email{j.szmidt@wil.waw.pl}
}

\date{Received: date / Accepted: date}
\maketitle
	
\begin{abstract}
The focus of this work is to show how to combine Zech's logarithms and each of the cycle joining and cross-join pairing methods to construct binary de Bruijn sequences of any order. A basic implementation is supplied as a proof-of-concept.

The cycles, in the cycle joining method, are typically generated by a linear feedback shift register. We prove a crucial characterization that determining Zech's logarithms is equivalent to identifying conjugate pairs shared by any two distinct cycles. This speeds up the task of building a connected adjacency subgraph that contains all vertices of the complete adjacency graph. Distinct spanning trees in either graph correspond to cyclically inequivalent de Bruijn sequences. As the cycles are being joined, guided by the conjugate pairs, we track the changes in the feedback function. Certificates of star or almost-star spanning trees conveniently handle large order cases.

The characterization of conjugate pairs via Zech's logarithms, as positional markings, is then adapted to identify cross-join pairs. A modified $m$-sequence is initially used, for ease of generation. The process can be repeated on each of the resulting de Bruijn sequences. We show how to integrate an analytic tool, attributed to Fryers, in the process.

Most prior constructions in the literature measure the complexity of the corresponding bit-by-bit algorithms. Our approach is different. We aim first to build a connected adjacency subgraph that is certified to contain all of the cycles as vertices. The ingredients are computed just once and concisely stored. Simple strategies are offered to keep the complexities low as the order grows.

\end{abstract}

\section{Introduction}\label{sec:intro}

A binary {\it de Bruijn sequence} of order $n \in \N$ has period $N=2^n$ in which each $n$-tuple occurs exactly once. Up to cyclic equivalence, there are $2^{2^{n-1}-n}$ of them~\cite{Bruijn46}. A naive approach capable of generating \emph{all} of these sequences, \eg, by generating all Eulerian cycles in the de Bruijn graph, requires an exponential amount of space. 

Adding a $0$ to the longest string of $0$s in a {\it maximal length sequence}, also known as an $m$-sequence, of period $2^n-1$ produces a de Bruijn sequence of order $n$. There are $\Lambda_{n}:=\frac{1}{n}\phi(2^n-1)$ such $m$-sequences where $\phi(\cdot)$ is the Euler totient function. There is a bijection between the set of all $m$-sequences and the set of primitive polynomials of degree $n$ in $\F_{2}[x]$. As $n$ grows large, $\Lambda_n$ soon becomes miniscule in comparison to $2^{2^{n-1}-n}$. In terms of applications, there is quite a diversity in the community of users. This results in varied criteria as to what constitute a good construction approach. 

Mainly for cryptographic purposes, there has been a sustained interest in efficiently generating a good number of de Bruijn sequences of large orders. For stream cipher applications (see, \eg, the 16 eSTREAM finalists in~\cite{Robshaw2008}), it is desirable to have high unpredictability, \ie, sequences with high nonlinearity, while still being somewhat easy to implement. This requires multiple measures of complexity to be satisfied. Due to their linearity, modified $m$-sequences are unsuitable for cryptographic applications.

Bioinformaticians want large libraries of sequences to use in experiments~\cite{Hume2014} and in genome assembly~\cite{CPT11}. Their design philosophy imposes some Hamming distance requirements, since a library must not have sequences which are too similar to one another or exhibiting some taboo patterns. Secrecy or pseudorandomness is often inconsequential. Designers of very large scale integrated circuits often utilize binary de Bruijn multiprocessor network~\cite{Samatham1989}. Important properties include the diameter and the degree of the network. The objective here is to have a small layout that still allows for many complete binary trees to be constructed systematically for extensibility, fault-tolerance, and self diagnosability.

\smallskip

\noindent
{\bf Prior Literature}

Given its long history and numerous applications, the quest to efficiently generate de Bruijn sequences has produced a large literature and remains very active. Most known approaches fall roughly into three types, each to be briefly discussed below. In terms of complexity, it is customary to specify the memory and time requirements to generate the next bit, \ie, the last entry in the next state given a current state. Such an algorithm is called a \emph{bit-by-bit algorithm}. 

As the name suggests, the {\it cycle joining} approach builds de Bruijn sequences by combining disjoint cycles of smaller periods \cite{Fred82}. One begins, for example, with a linear feedback shift register (LFSR) with an irreducible but nonprimitive characteristic polynomial of degree $n$. It produces disjoint cycles, each with a period that divides $2^n-1$. We identify a {\it conjugate pair} between any two distinct cycles and combine the cycles into a longer cycle. We repeat the process until all disjoint cycles have been combined into a de Bruijn sequence. The mechanism can of course be applied on {\it all} identifiable conjugate pairs.

Fredricksen applied the method on pure cycling register of length $n$ in \cite{Fred75}. The bit-by-bit routine took $6n$ bits of storage in $n$ units of time. Etzion and Lempel in \cite{EL84} considered a shift register that produced many short cycles which were then joined together. They proposed two constructions based on, respectively, pure cyclic registers (PCRs) and pure summing registers (PSRs). For a PCR, the number of produced sequence is exponential with base $2$ (see \cite[Theorem~2]{EL84} for the exact formula), with time complexity $\mathcal{O}(n)$ to produce the next bit. It was quite expensive in memory, requiring approximately $3n$ plus the exponent of $2$ in the exact formula. For PSR, the number of produced sequence is $\approx 2^{\frac{n^2}{4}}$ with its bit-by-bit algorithm costing $\mathcal{O}(n)$ in time and $\mathcal{O}(n^2)$ in memory. 

Huang showed how to join all the pure cycles of the complementing circulating register $f(x_1,\ldots,x_n)=\overline{x_1}$ in \cite{Huang90}. Its bit-by-bit algorithm took $4n$ bits of storage and $4n$ units of time. The resulting sequences tended to have a good local $0$ and $1$ balance. In the work of Jansen \etal \cite{JFB91}, the cycles were generated by an LFSR with feedback function $f(x)$ that can be factored into $r$ irreducible polynomials of the same degree $m$. The number of produced de Bruijn sequences was approximately $\mathcal{O}(2^{2n/\log(2n)})$. Its bit-by-bit complexity was $3n$ bits of storage and at most $4n$ FSR shifts. Numerous subsequent works, until very recently, are listed with the details on their input parameters and performance complexity in~\cite[Table 4]{Chang2019}. 
 
The {\it cross-join pairing} method begins with a known de Bruijn sequence, usually a modified $m$-sequence. One then identifies cross-join pairs that allow the sequence, seen as a cycle of period $2^n$, to be cut and reconnected into inequivalent de Bruijn sequences from the initial one. Two representative works on this method were done by Helleseth and Kl{\o}ve in~\cite{HK91} and by Mykkeltveit and Szmidt in~\cite{MS15}.

The third method, where we collect rather distinct approaches, is less clear-cut than the first two. The defining factor is that they follow some {\it assignment rules} or {\it preference functions} in producing one symbol at a time using the previous $n$ symbols. Some clarity can perhaps be given by way of several recent examples. 

A {\it necklace} is a string in its lexicographically smallest rotation. Stevens and Williams in~\cite{SW14} presented an application of the necklace prefix algorithm to construct de Bruijn sequences. Their loopless algorithm took $\mathcal{O}(n)$ time. They significantly simplified the FKM algorithm, attributed to Fredricksen, Kessler, and Maiorana, that concatenated the aperiodic prefixes of necklaces in lexicographic order. In \cite{SWW16}, Sawada \etal\, gave a simple shift-rule, \ie, a function that maps each length $n$ substring to the next length $n$ substring, for a de Bruijn sequence of order $n$. Its bit-by-bit algorithm ran in costant time using $\mathcal{O}(n)$ space. Instead of concatenating necklaces in lexicographic order, Dragon \etal\, in \cite{Dra18} used the co-lexicographic order to construct a de Bruijn sequence in $\mathcal{O}(n)$ time. Natural subsequences of the lexicographically least sequence had been shown to be de Bruijn sequences for interesting subsets of strings. The lexicographically-least de Bruijn sequence, due to Martin in 1934, is affectionately known as the {\it granddaddy}. 

The use of preference functions is prevalent in many greedy algorithms for binary de Bruijn sequence. A general greedy algorithm, recently proposed by Chang \etal\, in~\cite{CEF19}, included a discussion on how such functions fit into a unifying picture.

\smallskip
\noindent
{\bf Overview of our techniques}

We pinpoint the respective exact states in any conjugate pair shared by any two distinct cycles to execute the cycle joining process. At the core of our method is the use of Zech's logarithms to characterize the conjugate pairs as positional markings. The approach works particularly well when the characteristic polynomial of the LFSR is an irreducible polynomial or a product of distinct irreducible polynomials. One can identify and explicitly generate all of the component cycles with irreducible characteristic polynomial based on one special state.

We take a different route in the implementation. Instead of generating one de Bruijn sequence on the fly by using a bit-by-bit procedure, our mechanism ensures that users have a \emph{ready access to a large supply} of binary de Bruijn sequences. In its simplest form, the input consists of a primitive polynomial $p(x)$ of degree $n$ having $\alpha$ as a root and a positive integer $t$ that divides $2^n-1$ such that there exists a polynomial $q(x)$ of degree $n$ having $\beta = \alpha^t$ as a root. The cycle joining method is then applied to the cycles whose characteristic polynomial is $q(x)$. 

Preparatory computational tasks, that need to be done only once, supply the required ingredients to construct a connected adjacency subgraph $\widetilde{\mG}$. The graph contains all vertices of the complete adjacency graph $\mG$ whose vertices are the $t+1$ cycles. To build de Bruijn sequences from the stored $\widetilde{\mG}$, we can use Knuth's Algorithm S \cite[pp. 464 ff.]{Knuth4A} or an alternative algorithm \cite[Algorithm~5]{Chang2019}, due to Chang \etal\,, to generate all of the spanning trees. The two algorithms share the same complexity, as detailed in \cite[Section 8]{Chang2019}. If so desired, one can quickly choose a spanning tree at random by running Broder's Algorithm~\cite{Bro89}. Certificates that guarantee the existence of particular spanning trees in $\mG$ allow us to practically handle large orders, as will be demonstrated by examples. This simple strategy keeps the complexities low as the order grows.

We then show how to adapt finding conjugate pairs shared by two smaller cycles in the cycle joining method to finding two conjugate pairs in a given de Bruijn sequence to use as a cross-join pair, showing how closely connected the two methods are. Our constructions bypass the need for ``scaling,'' which has been a common theme in the studies of NLFSRs since its introduction by Lempel in~\cite{Lempel70}. Some notable recent investigations include the work of Dubrova in~\cite{Dub13} and that of Mandal and Gong in~\cite{MG13}.

Relevant notions and results on FSRs and Zech's logarithms are reviewed in Section~\ref{sec:prelims}, with new properties of the logarithms given towards the end. Section~\ref{sec3} characterizes conjugate pairs by Zech's logaritms based on a carefully chosen representation of the states. Section~\ref{sec:spantree} discusses the adjacency graph and methods to identify its spanning trees. A basic software impementation for the cycle joining method is given in Section~\ref{sec:program}. Section~\ref{sec:prod} integrates a recent proposal of using the product of irreducible polynomials as the characteristic polynomial of the starting LFSR into the current framework. Adapting the Zech's logarithms approach to identify cross-join pairs in an $m$-sequence is the core of Section~\ref{sec:cross}. We end with some conclusions and a general open problem in Section~\ref{sec:con}.

\section{Preliminaries}\label{sec:prelims}
For integers $k < \ell$, let $\bbra{k,\ell}$ denote $\{k,k+1,\ldots,\ell-1,\ell\}$. A {\it cyclotomic coset} of $2$ modulo $2^n-1$ containing $i$ is $D_i=\{i,2i,\ldots,2^{n_i-1} i\}$ with $n_i$ the least positive integer such that $i \equiv 2^{n_i} i \Mod{2^n-1}$. 
Obviously $n_i \mid n$. For each $i$, the least integer in $D_i$ is its coset leader. The set of all coset leaders is called a {\it complete set of coset representatives}, denoted by $\mathcal{R}_n$.

\subsection{Feedback Shift Registers}

An {\it $n$-stage shift register} is a clock-regulated circuit with $n$ consecutive units, each storing a bit. As the clock pulses, each bit is shifted to the next stage in line. A binary code is generated if one adds a feedback loop that outputs a new bit $s_n$ based on the $n$ bits $\s_0= (s_0,\ldots,s_{n-1})$ called an {\it initial state}. The codewords are the $n$-bit states. The Boolean function $h$ that outputs $s_n$ on input $\s_0$ is called its {\it feedback function}. The {\it algebraic normal form} (ANF) of $h$ is $\sum a_{i_1,i_2,\ldots,i_k} x_{i_1} x_{i_2} \cdots x_{i_k}$ with $a_{i_1,i_2,\ldots,i_k} \in \F_2$ and the sum is over all $k$-subsets $\{{i_1,i_2,\ldots,i_k}\} \subseteq \bbra{0,n-1}$. The {\it degree} of $h$ is the largest $k$ for which $a_{i_1,i_2,\ldots,i_k} \neq 0$.

A {\it feedback shift register} (FSR) outputs $\s=s_0,s_1,\ldots,s_n,\ldots$ satisfying the recursion $s_{n+\ell} = h(s_{\ell},s_{\ell+1},\ldots,s_{\ell+n-1})$ for $\ell \geq 0$. If $s_{i+N}=s_i$ for all $i \geq 0$, then $\s$ is {\it $N$-periodic}
or {\it with period $N$}, denoted by $\s= (s_0,s_1,s_2,\ldots,s_{N-1})$. The {\it $i$-th state} of $\s$ is $\s_i= (s_i,s_{i+1},\ldots,s_{i+n-1})$ and $\s_{i+1}$ is the {\it successor} of $\s_i$. We use $\0$ or $\0^k$ to denote the all zero vector or sequence of period $1$. The latter is preferred when the length $k$ needs to be specified.

Distinct initial states generate distinct sequences, forming the set $\Omega(h)$ of $2^n$ elements. A {\it state operator} $T$ turns $\s_i$ into $\s_{i+1}$, \ie, $\s_{i+1}=T \s_i$. 
If $\s$ has a state $\s_i$ and period $e$, the $e$ \emph{distinct} states of $\s$ are
$\s_i, T \s_i = \s_{i+1},\ldots, T^{e-1} \s_i = \s_{i+e-1}$. The {\it shift operator} $L$ sends $\s$ to $L\s=L(s_0,s_1,\ldots,s_{N-1})=(s_1,s_2,\ldots,s_{N-1},s_0)$, 
with the convention that $L^0\s=\s$. The set $[\s]:=\{\s,L\s,L^2\s,\ldots,L^{N-1}\s \}$ 
is a {\it shift equivalent class} or a {\it cycle} in $\Omega(h)$. 
One partitions $\Omega(h)$ into cycles and writes the {\it cycle structure} as
\begin{equation}\label{eq:cycle}
[\s_1] \cup [\s_2] \cup \ldots \cup [\s_r] \mbox{ if } \bigcup_{i=1}^{r} [\s_i] = \Omega(h).
\end{equation}

An FSR with a linear feedback function is {\it linear} (or an LFSR) and {\it nonlinear} (or an NLFSR) otherwise. The {\it characteristic polynomial} of an $n$-stage LFSR is 
\begin{equation}\label{equ:cp}
f(x)=x^n+c_{n-1}x^{n-1}+\ldots+c_1x+c_0\in \F_2[x]
\end{equation}
when the feedback function is $h(x_0,\ldots,x_{n-1})= \sum_{i=0}^{n-1} c_i x_i$. To ensure that all generated sequences are periodic, $c_0 \neq 0$. 
A sequence $\s$ may have many characteristic polynomials. The one with the 
lowest degree is its {\it minimal polynomial}. It represents the LFSR of 
shortest length that generates $\s$. Given an LFSR with characteristic 
polynomial $f(x)$, the set $\Omega(h)$ is also denoted by $\Omega(f(x))$. 
To $f(x)$ in (\ref{equ:cp}), one associates a matrix
\begin{equation}\label{comx}
A_{f}:=\begin{pmatrix}
0 & 0  & \ldots & 0 & c_0 \\
1 & 0  & \ldots & 0 & c_1 \\
0 & 1  & \ldots & 0 & c_2 \\
\vdots & \vdots & \ddots & \vdots & \vdots\\
0 & 0  & \ldots & 1 & c_{n-1}
\end{pmatrix}.
\end{equation}
On input state $\s_0$, the state vectors of the resulting sequence are $\s_{j}=\s_0 A^{j}$ for $j \in \{0,1,2,\ldots\}$ and $\s_{i+1}=\s_iA=T\s_i$.

If the minimal polynomial of a sequence is primitive with degree $n$ and $\alpha$ as a root, then it can be obtained by applying the trace map from $\F_{2^n}$ onto $\F_2$ on the sequence $1,\alpha,\alpha^2,\ldots,\alpha^{2^n-2}$. The sequence is said to be a {\it maximal length sequnce} or an $m$-sequence since it has the maximal period among all sequences generated by any LFSR with minimal polynomial of degree $n$. Let $\m$ denote an $m$-sequence. A sequence $\u$ is a {\it $d$-decimation} sequence of $\s$, denoted by $\u=\s^{(d)}$ if $u_j=s_{d \cdot j}$ for all $j \geq 0$. A $d$-decimation $\m^{(d)}$ of $\m$ is also an $m$-sequence if and only if $\gcd(d,2^n-1)=1$. More properties of sequences in relation to their characteristic and minimal polynomials can be found in~\cite[Chapter~4]{GG05} and \cite[Chapter~8]{LN97}.

\subsection{Zech's Logarithms}\label{subsec:zech}

Among computational tools over finite fields we have {\it Zech's logarithm} (see \eg,~\cite[page~39]{GG05}). The logarithm is often referred to as {\it Jacobi's logarithm}~\cite[Exercise 2.8 and Table B]{LN97}. It was Jacobi who introduced the notion and tabulated the values for $\F_p$ with $p \leq 103$ in~\cite{Jacobi}.
For $\ell \in \bbra{1,2^n-1}\cup\{-\infty\}$, the Zech's logarithm $\tau(\ell)$ relative to $\alpha$ is defined by $1+\alpha^{\ell}=\alpha^{\tau(\ell)}$ where $\alpha^{-\infty}=0$. It induces a permutation on $\bbra{1,2^n-2}$. Huber established the following properties of Zech's logarithms over $\F_q$ and provided the Zech's logarithm tables for $\F_{2^n}$ with $2 \leq n \leq 11$ in~\cite{Huber90}. 

\begin{enumerate}
\item It suffices to find the logarithms relative to one primitive element. 
Let distinct primitive polynomials $p(x)$ and $q(x)$ be of degree $n$ with respective roots $\alpha$ and $\delta$. 
Then $\delta=\alpha^b$ for some $b$ with $\gcd(b,2^n-1)=1$. Let $\tau_{n,\alpha}(k)$ and $\tau_{n,\delta}(k)$ 
denote the respective logarithms of $k \in \bbra{1,2^n-2}$ relative to $\alpha$ and $\delta$. Then 
\[
1+\delta^k=1+\alpha^{b \cdot k}=\alpha^{\tau_{n,\alpha}(b \cdot k)}=
\alpha^{bb^{-1}\tau_{n,\alpha}(b \cdot k)}=\delta^{b^{-1}\tau_{n,\alpha}(b \cdot k)}.
\]
Hence, $\tau_{n,\delta}(k)\equiv b^{-1}\tau_{n,\alpha}(b \cdot k)\Mod{(2^n-1)}$. With a primitive element fixed, the notation is $\tau(k)$ or $\tau_n(k)$ to emphasize $n$.
	
\item The $\flip$ map is given by $\tau_n(\tau_n(k))=k$. Knowing $\tau_n(k)$ for any $k$ in a cyclotomic coset $D_j$ is sufficient to find $\tau_n(2^i k)$ by using the $\double$ map
\begin{equation}\label{eq:double}
\tau_n(2k)\equiv 2 \tau_n(k)\Mod{(2^n-1)}.
\end{equation}
The $\flip$ and $\double$ maps send cosets onto cosets of the same size. 
	
\item Let ${\In}(j)=2^n-1-j \equiv -j \Mod{(2^n-1)}$. Then
\begin{equation}\label{tau2}
\tau_n(\In(j))\equiv \tau_n(j)-j \Mod{(2^n-1)}.
\end{equation}
$\In(k)$ maps a coset onto a coset of the same size. 
	
\item Let $h(x)$ be a primitive polynomial of degree $m$ having a root 
$\beta$. Let $m \mid n$ and $\beta=\alpha^r$ with $r=\frac{(2^n-1)}{(2^m-1)}$. If the Zech's logarithms relative to $\beta$ are known, then $1+\alpha^{r \cdot j}=1+\beta^j =\beta^{\tau_m(j)}=\alpha^{r \cdot \tau_m(j)}$.
Hence,
\begin{equation}\label{tau3}
\tau_n(r \cdot j) \equiv r \cdot \tau_m(j) \Mod{(2^n-1)}.
\end{equation}
\end{enumerate}

Repeatedly applying $\flip$ and $\In$ induces a cycle of $6$ cosets, except in rare cases (see~\cite{Huber90}). Using the $\double$ map, one then gets the value of $\tau_n(k)$ for all $k$ in the union of these cosets. To complete the table, Huber suggested the use of {\it key elements}, each corresponding to a starting coset. Instead of following the suggestion, we introduce a more efficient route that reduces the storage need.

\begin{lemma}\label{lem:tau}
For known $\tau_n(i)$, $\tau_n(j)$, and $\tau_n(i+j)$,
\begin{equation}\label{tau4}
\tau_n(\tau_n(i+j)-\tau_n(j))\equiv \tau_n(i) + j -\tau_n(j) \pmod{(2^n-1)}.
\end{equation}
\end{lemma}

\begin{proof}
For $i,j \in \bbra{1,2^n-2}$, $1+\alpha^i =\alpha^{\tau_n(i)}$ if and only if $\alpha^{j}+\alpha^{i+j}=\alpha^{\tau_n(i)+j}$. This is equivalent to $\alpha^{\tau_n(j)}+\alpha^{\tau_n(i+j)}=\alpha^{\tau_n(i)+j}$. 
Thus, $1+\alpha^{\tau_n(i+j)-\tau_n(j)}=\alpha^{\tau_n(i)+j-\tau_n(j)}$. \qed
\end{proof}
To apply Lemma~\ref{lem:tau}, one looks for an $(i,j)$ pair such that the respective Zech's logarithms of $i,j$, and $i+j$ are already known, \ie, the elements $i,j$, and $i+j$ are in the union $U$ of cosets with known Zech's logarithms, but $\tau(i) - \tau(j) \notin U$. In many cases, a given Zech's logarithm is sufficient to deduce all others values. We can write Equation (\ref{tau4}) as 
$\tau_n(\tau_n(i)-\tau_n(j))\equiv \tau_n(i-j) + j -\tau_n(j) \pmod{(2^n-1)}$.

\begin{example}\label{ex:ZechTable}
We reproduce the table in~\cite[Appendix]{Huber90} for $p(x)=x^{10}+x^3+1$ 
without any key element. The computations are done modulo $2^{10}-1$ with $=$ replacing $\equiv$ for brevity. There are $107$ cyclotomic cosets of $2$ modulo $1023$: the trivial coset $D_0$, a coset of cardinality $2$, $6$ cosets, each of cardinality $5$, and $99$ cosets, each of cardinality $10$. The coset $\{341,682\}$ implies $\tau(341)=682$. The cycle of $6$ cosets beginning from $\tau(3)=10$ is 
$3 \in D_3 ~  \lra{\flip} ~ 10 \in D_5 ~ \lra{\In} ~ 1013 \in D_{383} ~\lra{\flip} ~ 1016 \in D_{127} \lra{\In} ~ 7 \in D_7 ~ \lra{\flip} ~ 1020 \in D_{255}~\lra{\In}~ 3 \in D_3$, giving the logarithms of all $60$ elements in $\bigcup D_{k} \mbox{ with } k \in \{3,5,7,127,255,383\}$. 
For the remaining logarithms, one searches for an $(i,j)$ pair with $\tau(i)-\tau(j)$ not in any of previously known cosets but $i-j$ is. The task for the remaining cosets is summarized in Table~\ref{table:cosets}. The rows follow chronological order. 

The approach may fail to yield the complete table. We tested all trinomials $x^n+x^j+1$ with $n \in \{15,17,18,20,22\}$ and $j \leq \floor{n/2}$ since the reciprocals give identical conclusions. When $n=15$, the full table is obtained for $j \in \{1,4,7\}$. When $n=17$, Lemma~\ref{lem:tau} produces the full table for $j \in \{3,5\}$ but fails for $j=6$. It also fails for $n=18,j=7$ but works for $n=20,j=3$ and $n=22,j=1$. In case of failure, the computation in (\ref{tau3}) becomes necessary. \qed
	
\begin{table}[h!]
\caption{Zech's Logarithms for Elements in Remaining Cosets of Example~\ref{ex:ZechTable}.}
\label{table:cosets}
\renewcommand{\arraystretch}{1.1}
\centering
\begin{tabular}{lllc}
\hline
$(i,j)$ & $\tau(\tau(i)-\tau(j))$ & Cosets $D_{k}$s in the induced cycle  & $\abs{\cup_{k} D_k }$ \\
\hline
$(12,5)$ & $\tau(550)=512$ & $ k \in \{77,1,511,19,251,187\}$ & $60$\\
$(12,7)$ & $\tau(43)=523$ & $ k \in \{43,23,125,63,15,245\}$ & $60$\\
$(76,28)$ & $\tau(11)=200$ & $ k \in \{11,25,223,45,189,253\}$ & $60$\\
$(3,1)$ & $\tau(956)=78$ & $ k \in \{239,39,123,439,73,49\}$ & $60$\\
$(3,2)$ & $\tau(879)=948$ & $ k \in \{447,237,75,375,69,9\}$ & $60$\\
			
$(12,2)$ & $\tau(909)=874$ & $ k \in \{111,347,149,35,247,57\}$ & $60$\\
$(12,10)$ & $\tau(37)=161$ & $ k \in \{37,379,31\}$ & $30$\\
$(37,31)$ & $\tau(426)=316$ & $ k \in \{213,79,59,55,121,171\}$ & $60$\\
$(37,6)$ & $\tau(141)=744$ & $ k \in \{93,105,183\}$ & $30$\\
$(77,43)$ & $\tau(501)=142$ & $ k \in \{351,71,119,235,83,21\}$ & $60$\\
			
$(77,34)$ & $\tau(402)=958$ & $ k \in \{147,479,17,221,89,219\}$ & $60$\\
$(68,12)$ & $\tau(181)=971$ & $ k \in \{181,191,13,157,91,173\}$ & $60$\\
$(749,255)$ & $\tau(29)=566$ & $ k \in \{29,109,151,207,51,95\}$ & $60$\\
$(702,136)$ & $\tau(343)=746$ & $ k \in \{343,85,155\}$ & $30$\\
$(434,109)$ & $\tau(27)=206$ & $ k \in \{27,103,115,205,179,159\}$ & $60$\\
			
$(349,333)$ & $\tau(33)=660$ & $ k \in \{33,165,363,99,231,495\}$ & $30$\\
$(785,151)$ & $\tau(87)=619$ & $ k \in \{87,215,117,367,101,41\}$ & $60$\\
$(274,51)$ & $\tau(107)=376$ & $ k \in \{107,47,175,61,167,53\}$ & $60$\\
\hline
\end{tabular}
\end{table}
\end{example}

\section{Zech's Logarithms and Conjugate Pairs}\label{sec3}

A state $\v:=(v_0,v_{1},\ldots,v_{n-1})$ and its {\it conjugate} $\overline{\v}:=(v_0 + 1,v_{1},\ldots,v_{n-1})$ form a {\it conjugate pair}. Cycles $C_1$ and $C_2$ in $\Omega(h)$ are {\it adjacent} if they are disjoint and there exists $\v$ in $C_1$ whose conjugate $\overline{\v}$ is in $C_2$. Adjacent cycles merge into a single cycle by interchanging the successors of $\v$ and $\overline{\v}$. The feedback function of the resulting cycle is
\begin{equation}\label{eq:CJ_feed}
\widetilde{h}(x_0,\ldots,x_{n-1})=h(x_0,\ldots,x_{n-1})+\prod_{i=1}^{n-1}(x_i+v_i+1).
\end{equation}
Continuing this step, all cycles in $\Omega(h)$ join into a de Bruijn sequence. The feedback functions of the resulting de Bruijn sequences are completely determined once the corresponding conjugate pairs are found. 

\subsection{Basic Notions and Known Results}

\begin{definition}(From~\cite{HM96}\label{def:adjgraph}) The {\it adjacency graph} $\mG$ of an FSR with feedback function $h$ is an undirected multigraph whose vertices correspond to the cycles in $\Omega(h)$. There exists an edge between two vertices if and only if they are adjacent. A conjugate pair labels every edge. The number of edges between any pair of cycles is the number of conjugate pairs that they share.
\end{definition}
Clearly $\mG$ has no loops. There is a bijection between the spanning trees of $\mG$ and the de Bruijn sequences constructed by the cycle joining method (see. \eg, ~\cite{HM96} and~\cite{HH96}). The following well-known counting formula is a variant of the BEST (de {\bf B}ruijn, {\bf E}hrenfest, {\bf S}mith, and {\bf T}utte) Theorem from~\cite[Section~7]{AEB87}. Recall that the {\it cofactor} of entry $m_{i,j}$ in $\mM$ is $(-1)^{i+j}$ times the determinant of the matrix obtained by deleting the $i$-th row and $j$-th column of $\mM$.

\begin{theorem}(BEST)\label{BEST} Let $V_{\mG}:=\{V_1,V_2,\ldots,V_k\}$ be the vertex set of the adjacency graph $\mG$ of an FSR. Let $\mM=(m_{i,j})$ be the $k \times k$ matrix derived from $\mG$ in which $m_{i,i}$ is the number of edges incident to vertex $V_i$ and $m_{i,j}$ is the negative of the number of edges between vertices $V_i$ and $V_j$ for $i \neq j$. Then the number of the spanning trees of $\mG$ is the cofactor of any entry of $\mM$.
\end{theorem}

To generate de Bruijn sequences of a large order, efficient identification of conjugate pairs is crucial. A de Bruijn sequence generator has been recently proposed in~\cite{Chang2019}. Its basic software implementation \cite{EF18} demonstrates a decent performance up to order $n \approx 20$ when the characteristic polynomials are products of distinct irreducible polynomials. As the order grows the program soon runs into time complexity issues. Another recent contribution can be found in~\cite{Dong16}. The work shows how to generate de Bruijn sequences of large orders, say $128$, but without supplying time or space complexity analysis. The above examples indicates that dealing with large order demands much faster algorithmic tools.

Let $\mu(n)$ be the M{\"o}bius function. There are $\frac{1}{n} \sum_{d \mid n} \mu(d) 2^{\frac{n}{d}}$ binary irreducible polynomials of degree $n$. All irreducible polynomials of (prime) degree $n$ are primitive if and only if $2^n-1$ is prime. Such an $n$ is called a {\it Mersenne exponent}. Although not known to be finite, Mersenne exponents are sparse~\cite{OEIS43}. Thus, for most $n$, there are many more irreducible than primitive polynomials. Each such polynomial yields a number of de Bruijn sequences if one can efficiently identify the conjugate pairs. We show how to accomplish this task.

A primitive polynomial $p(x) \in \F_2[x]$ of degree $n$ having a root $\alpha$ can be identified in several ways, (see, \eg, ~\cite[Section 4.4]{Arn10}). Many computer algebra systems have routines that output primitive polynomial(s) of a specified degree. Combining a decimation technique and the Berlekamp-Massey algorithm on input $p(x)$ and a suitable divisor $t$ of $2^n-1$ yields the associated irreducible polynomial $f(x)$. It has degree $n$, order $e$, and a root $\beta=\alpha^t$ with $e \cdot t = 2^n-1$. Notice that $n$ is the least integer satisfying $2^n \equiv 1 \Mod{e}$. 

The $t$ distinct nonzero cycles in $\Omega(f(x))$ can be ordered in the language of cyclotomic classes and numbers using a method from~\cite{HH96}. The {\it cyclotomic classes} $C_i\subseteq\F_{2^n}^{*}$, for $0 \leq i <t$, are
\begin{equation}\label{eq:cyclas}
C_i=\{\alpha^{i+ s \cdot t}~|~0\leq s <e\}=\{\alpha^i\beta^s~|~0\leq s<e\}=\alpha^i C_0.
\end{equation}
The {\it cyclotomic numbers} $(i,j)_{t}$, for $0\leq i,j <t$, are
\begin{equation}\label{eq:cycnum}
(i,j)_{t} =\left|\{\xi~|~\xi\in C_i, \xi+1\in C_j\}\right|.
\end{equation}
Using the basis $\{1,\beta,\ldots,\beta^{n-1}\}$ for $\F_{2^n}$ we write $\displaystyle{\alpha^j=\sum_{i=0}^{n-1}a_{j,i} \beta^i}$
with $a_{j,i} \in \F_2$ for $j \in \bbra{0,2^n-2}$. In vector form, the expression becomes
\begin{equation}\label{eq:vecform}
\alpha^j=(a_{j,0},a_{j,1},\ldots,a_{j,n-1}).
\end{equation}
Define the mapping $\varphi:\F_{2^n}\rightarrow \F_2^{n}$ by
\begin{equation}\label{eq:varphi}
\varphi(0)=\0,\ \ \varphi(\alpha^{j})=(a_{j,0},a_{j+t,0},\ldots,a_{j+(n-1)t,0}),
\end{equation}
where the subscripts are reduced modulo $2^n-1$.
By the recursive relation determined by (\ref{eq:vecform}), $\varphi$ is a bijection. Let
\begin{equation}\label{eq:corres}
\u_i := (a_{i,0},a_{i+t,0},\ldots,a_{i+(e-1)t,0}).
\end{equation}
It is now straightforward to verify that 
\begin{equation}\label{equ:g}
\Omega(f(x))=[\0]\cup[\u_0]\cup[\u_1]\cup\ldots\cup[\u_{t-1}]
\end{equation}
with $\varphi(\alpha^{i})$ as the initial state of $\u_i$ for $i \in \bbra{0,t-1}$. 
In particular, the initial state of $\u_0$ is $(1,\0) \in \F_2^n$. Note that $\varphi$ induces a correspondence between $C_i$ and $[\u_i]$ (see~\cite[Theorem 3]{HH96}). 
In other words, $\u_i$ and the sequence of states of $\u_i$,
namely $((\u_i)_0,(\u_i)_1,\ldots,(\u_i)_{e-1})$, where
\[
(\u_i)_j=(a_{i+jt,0},a_{i+(j+1)t,0},\ldots,a_{i+(j+n-1)t,0})
=\varphi(\alpha^i\beta^j)
\]
for $j \in \bbra{0,e-1}$,
are equivalent. The state $\varphi(\alpha^i\beta^j)$ corresponds to the
element $\alpha^i\beta^j \in C_i$. Hence, $\u_i \longleftrightarrow C_i$. 
This provides a convenient method to find the exact position of any state 
$\v=(v_0,v_1,\ldots,v_{n-1})\in \F_2^n$ in some cycle in $\Omega(f(x))$. 

\begin{proposition}\cite[Corollary 2]{HH96}\label{prop:from HH}	
The states $\varphi(\theta)$ and $\varphi(\theta+1)$ form a conjugate pair for all $\theta \in \F_{2^n}$.
\end{proposition}

\subsection{Efficient Exact Determination of Conjugate Pairs}

We extend Proposition~\ref{prop:from HH} by explicitly determining the respective exact positions of the states in their corresponding cycles. 
This is far from trivial. First, one needs to fix a {\it standard} representation of the cycles and then uses algebraic tools to find the very location of $\varphi(\theta)$ in $C_1$ and $\varphi(\theta+1)$ in $C_2$. 

Let $\v=\varphi(\alpha^i\beta^j)=
\varphi(\alpha^{i+tj})$ for some $i \in \bbra{0,t-1}$. Then $\v$ must be the $j$-th state of $\u_i$. Let $k=i+tj$ and suppose that 
$\alpha^k=a_0+a_1\beta+\ldots+a_{n-1}\beta^{n-1}$. We have $a_0=v_0$ from the definition of $\varphi(\alpha^k)$. 
Note that if $f(x)$ in (\ref{equ:cp}) is irreducible, then $c_0=1$. Since $\beta=\alpha^t$,
\begin{multline*}
\alpha^{k+t} =\sum_{\ell=0}^{n-1}a_{\ell}\beta^{\ell+1}= 
\sum_{\ell=0}^{n-2}a_{\ell} \beta^{\ell+1}+\sum_{\ell=0}^{n-1}a_{n-1} c_{\ell} \beta^{\ell} = \notag\\
a_{n-1}+(a_0+a_{n-1}c_1)\beta+\ldots+(a_{n-2}+a_{n-1}c_{n-1})\beta^{n-1}.
\end{multline*}
Continuing inductively, it becomes clear that $v_j$ is the constant term in the linear combination of $\alpha^{k+tj}$ in the $\beta$ basis. Hence,
\[
\begin{cases}
a_0    & =v_0 \\
a_{n-1}& =v_1\\
a_{n-2}& =v_1 c_{n-1}+v_2\\
a_{n-3}& =v_1 c_{n-2}+v_2 c_{n-1}+v_3\\
\vdots & \vdots \\
a_1    & =v_1 c_2+\ldots+v_{n-2} c_{n-1}+v_{n-1}
\end{cases}.
\]
Once $a_0+a_1\beta+\ldots+a_{n-1}\beta^{n-1}$ and the Zech's logarithms relative to $\alpha$ are known, one gets $\alpha^k$ and, thus, 
$\v$'s position.

How can one efficiently generate the cycles in $\Omega(f(x))$? Directly using the above definition is not practical since it requires 
costly computations over $\F_{2^n}$. Simply generating them by the LFSR with characteristic polynomial $f(x)$ may fail to 
guarantee that their respective initial states are $\varphi(\alpha^i)$ for $i \in \bbra{0,t-1}$. We show that decimation is the right tool.

Equation (\ref{eq:vecform}) ensures that $\m=(a_{0,0},a_{1,0},\ldots,a_{2^n-2,0})$ is an $m$-sequence 
with characteristic polynomial $p(x)$~\cite[Chapter 5]{GG05}.
The {\it trace function} $\Tr$ maps $\delta \in \F_{2^n}$ to $\sum_{i=0}^{n-1} \delta^{2^i} \in \F_2$. Recall, \eg, from~\cite[Section 4.6]{GG05} that the entries in $\m$ are $a_{i,0}=\Tr(\gamma \alpha^i): 0 \neq \gamma \in \F_{2^n}$ for $i \in \bbra{0,2^n-2}$.
From $\m$, construct $t$ distinct $t$-decimation sequences of period $e$:
\[
\u_0=\m^{(t)},\u_1=(L\m)^{(t)},\ldots,\u_{t-1}=(L^{t-1}\m)^{(t)}.
\]
The resulting sequences have $(\u_k)_j= \Tr(\gamma\alpha^{k+ t\cdot j})$ for $k \in \bbra{0,t-1}$ and $j \in \bbra{0,e-1}$. 
Each $[\u_i]$ is a cycle in $\Omega(f(x))$ since $\beta=\alpha^t$. 
We need to find an initial state $\v$ of $\m$ such that the initial state of $\u_0$ is $(1,\0) \in \F_2^n$.

Let $A_{p}$ be the associate matrix of $p(x)$. Then the respective first elements of $\v A_{p}^{(i \cdot t)}$ for $t \in \bbra{0,n-1}$ must be $1,0,\ldots,0$. A system of equations to derive $\v$ can be constructed. Let $\kappa$ be the number of $1$s in the binary representation 
of $i \cdot t$. Computing $A_{p}^{(i \cdot t)}$ is efficient using the square-and-multiply method, taking at most 
$\log_2 \floor{i \cdot t}$ squarings and $\kappa$ multiplications.

We use $\v$ and $p(x)$ to derive the first $n \cdot t$ entries of $\m$. The respective initial states $\varphi(\alpha^0),\ldots,\varphi(\alpha^{t-1})$ of $\u_0,\ldots,\u_{t-1}$ in $\Omega(f(x))$ immediately follow by decimation. Thus, one gets $\varphi(\alpha^j)$ for any $j$.
This allows us to quickly find the desired initial state of any cycle, even for large $n$.
Given the state $(\u_i)_j=\varphi(\alpha^i\beta^j)$, we have  $T^k[(\u_i)_j]=\varphi(\alpha^i\beta^{j+k})$. At this point, given an irreducible polynomial $f(x)$ with root $\beta$ and order $e=\frac{2^n-1}{t}$, 
we need a primitive polynomial $p(x)$ with the same degree as $f(x)$ and a root $\alpha$ satisfying $\beta=\alpha^t$. In general, such a $p(x)$ is not unique~\cite[Subsection 3.2]{Chang2019}. Here we provide a method to find one.

For $k \in \bbra{1,\Lambda_{n}}$,
let $p_k(x)$ be a primitive polynomial of degree
$n$ that generates the $m$-sequence $\m_k$. The set of all
shift inequivalent $m$-sequences with period $2^n-1$ is
$\{\m_1,\m_2,\ldots,\m_{\Lambda_{n}}\}$. The elements are the $d_j$-decimation sequences of any $\m_k$ for all $d_j$ satisfying $\gcd(d_j,2^n-1)=1$. We derive $\m_{k}^{(t)}$ of period $e$ and check if it shares a common string of $2n$ consecutive elements with a sequence whose characteristic polynomial is $f(x)$. If yes, then we associate $p_k(x)$ with $f(x)$. Testing all $k$s guarantees a match between $f(x)$ and some $p_k(x)$ without 
costly operations over $\F_{2^n}$.

As $n$ or $t$ grows, finding one such $p(x)$ becomes more computationally involved. To work around this limitation, one starts instead with any primitive polynomial $p(x)$ with a root $\alpha$ and find the corresponding irreducible $f(x)$ having $\beta=\alpha^t$ as a root. There are tools from finite fields (see, \eg,~\cite{LN97})
that can be deployed. We prefer another approach that does not require computations over $\F_{2^n}$.

Any primitive polynomial $p(x)$ generates an $m$-sequence $\m$. By $t$-decimation we get $\m^{(t)}$. We input any $2n$ consecutive bits of $\m^{(t)}$ into the Berlekamp-Massey Algorithm~\cite[Section 6.2.3]{Menezes:1996} to get an irreducible
polynomial $f(x)$ having $\alpha^t$ as a root. There are instances where
$f(x)$ has degree $m \mid n$ with $m < n$. This implies that, for this $t$, there is no irreducible polynomial of degree $n$ that can be associated with $p(x)$. As $k$ traverses $\mathcal{R}_n$, by $k$-decimation and the Berlekamp-Massey algorithm, the process provides all irreducible polynomials with root $\alpha^{k}$. The resulting polynomials form the set of all irreducible polynomials with degree $|C_k|=m \mid n$. For $t \in \mathcal{R}_n$, $t$ is valid to use if and only if $|C_t|=n$.

\begin{proposition}\label{prop:commdiag}
For a valid $t$, let the $m$-sequence $\m_2$, whose characteristic polynomial is $p_2(x)$, be the $d_j$-decimation sequence of the $m$-sequence $\m_1$, which is generated by $p_1(x)$. Let $\Psi_{d_j}$ be the mapping, induced by the decimation, that sends $p_1(x)$ to $p_2(x)$. Let $p_i(x)$ be the respective associated primitive polynomial of the irreducible polynomial $q_i(x)$, both of degree $n$. The mapping $\Gamma_t$ is induced by the mapping $\alpha_i \mapsto \beta_i =\alpha_i^t$. Then we have the commutative diagram 
\[ 
\begin{tikzcd}
p_1(x) \arrow{r}{\Psi_{d_j}} \arrow[swap]{d}{\Gamma_t} & 
p_2(x) \arrow{d}{\Gamma_t} \\%
q_1(x) \arrow{r}{\Psi_{d_j}}& q_2(x)
\end{tikzcd}.
\]	
\end{proposition}

\begin{proof}
Let $\alpha_i$ be a primitive root of $p_i(x)$. Hence, there exists a root $\beta_i$ of $q_i(x)$ such that $\beta_i=\alpha_i^t$. Applying the trace map $\Tr$ shows that the two directions coincide by the commutativity of the exponents in a multiplication. \qed
\end{proof}

We can now characterize the conjugate pairs shared by any two distinct cycles by Zech's logarithms. In fact, determining the respective initial states of $[\u_i] : i \in \bbra{1,t-1}$ is not even necessary since ensuring that $(1,\0)$ is the initial state of $[\u_0]$ is sufficient for implementation. 
\begin{theorem}\label{thm:equivalence}
Let $\alpha$ be a root of a primitive polynomial $p(x)$ of degree $n$ and $\tau()$ be the Zech's logarithm with respect to $\alpha$. Let $f(x)$ be the irreducible polynomial of degree $n$ and order $e=\frac{2^n-1}{t}$ having a root $\beta=\alpha^t$, \ie, $f(x)$ is associated with $p(x)$. 

Let $[\u_i]$ and $[\u_{\ell}]$ be distinct nonzero cycles in $\Omega(f(x))$ constructed above, \ie, $i,\ell \in \bbra{0,t-1}$ with $i \neq \ell$. Let $\v:=T^j\varphi(\alpha^i)=\varphi(\alpha^{i+tj})$ be the $j$-th state of $[\u_i]$ and $\overline{\v}:=T^k \varphi(\alpha^{\ell})=\varphi(\alpha^{\ell+tk})$ be the $k$-th state of $[\u_{\ell}]$. Then $(\v,\overline{\v})$ forms a conjugate pair if and only if $\ell + tk = \tau(i+tj)$.
\end{theorem}

\begin{proof}
Let $\eta$ and $\gamma$ be elements of $\F_{2^n}$. 
Then $\varphi(\eta)$ is a state of $\u_i$ if and only if 
$\eta=\alpha^j$ and $j \in C_i$. It is therefore clear that $\varphi(\eta)+\varphi(\gamma)=\varphi(\eta+\gamma)$. 
Observe that $\varphi(\alpha^0=1)$ and $\0$ are conjugate. 
The conjugate of $\varphi(\alpha^j)$ with $j \in \bbra{1,2^n-2}$ is
\[
\overline{\varphi(\alpha^j)}=\varphi(1) + \varphi(\alpha^j)
=\varphi(1+\alpha^j)=\varphi(\alpha^{\tau(j)}) \mbox{ where } 
\varphi(1)=(1,0,\ldots,0).
\]
The conjugate of any state $\varphi(\alpha^j)$ in cycle $[\u_{j\Mod{t}}]$ must then be $\varphi(\alpha^{\tau(j)})$ in cycle $[\u_{\tau(j)\Mod{t}}]$.
In other words, the conjugate of the $j$-th state of cycle $[\u_i]$, which is 
$T^j\varphi(\alpha^i)=\varphi(\alpha^i\beta^j)=\varphi(\alpha^{i+tj})$, must be $\varphi(\alpha^{\tau(i+tj)})$. 
Writing $\tau(i+tj)=kt+ \ell$ with 
$k \in \bbra{0,e-1}$ and $\ell \in \bbra{0,t-1}$, 
$\varphi(\alpha^{\tau(i+tj)})=T^{k}\varphi(\alpha^{\ell})$ 
belongs to $[\u_{\ell}]$. 

Thus, knowing the Zech's logarithms relative to $\alpha$ enables us to 
easily determine all conjugate pairs between two arbitrary cycles in $\Omega(f(x))$. 
By the definition of cyclotomic numbers, $[\u_i]$ and $[\u_j]$ share $(i,j)_t$ conjugate pairs. 

Conversely, knowing all of the conjugate pairs allows us to derive the Zech's logarithms relative to $\alpha$.
Let a conjugate pair $(\v,\overline{\v})$ with $\v=T^j\varphi(\alpha^i)=\varphi(\alpha^{i+tj})$
and $\overline{\v}=T^k\varphi(\alpha^{\ell})=\varphi(\alpha^{\ell+tk})$ be given. Then $\tau(i+tj)=\ell+tk$ since 
$\overline{\v}$ must be $\varphi(\alpha^{\tau(i+tj)})$. If all of the conjugate pairs are known, a complete Zech's 
logarithm table, relative to $\alpha$, follows. \qed
\end{proof}

\begin{example}\label{example1}
Given $f(x)=x^4+x^3+x^2+x+1$, which is irreducible, of order $5$ with $\beta$ as a root, 
choose $p(x)=x^4+x+1$ with a root $\alpha$ satisfying $\alpha^3=\beta$ as the associated 
primitive polynomial. Let $\m$ be the corresponding $m$-sequence with initial state $(1,0,0,0)$. By $3$-decimating one derives $\Omega(f(x))=[\0]\cup[\u_0]\cup[\u_1]\cup[\u_2]$ with $\u_0=(1,0,0,0,1)$, $\u_1=(0,1,1,1,1)$, and $\u_2=(0,0,1,0,1)$. The nonzero $4$-stage states are

\begin{center}
\begin{tabular}{l | l}
$\varphi(\alpha^0 =\alpha^0\beta^0)=(1,0,0,0)=(\u_0)_0$ &  
$\varphi(\alpha^8 =\alpha^2\beta^2)=(1,0,1,0)=(\u_2)_2$\\

$\varphi(\alpha^1 = \alpha^1\beta^0)=(0,1,1,1)=(\u_1)_0$ &
$\varphi(\alpha^9 = \alpha^0\beta^3)=(0,1,1,0)=(\u_0)_3$ \\

$\varphi(\alpha^2 = \alpha^2\beta^0)=(0,0,1,0)=(\u_2)_0$ &
$\varphi(\alpha^{10} = \alpha^1\beta^3)=(1,1,0,1)=(\u_1)_3$\\

$\varphi(\alpha^3 = \alpha^0\beta^1)=(0,0,0,1)=(\u_0)_1$ &
$\varphi(\alpha^{11} = \alpha^2\beta^3)=(0,1,0,0)=(\u_2)_3$\\

$\varphi(\alpha^4 = \alpha^1\beta^1)=(1,1,1,1)=(\u_1)_1$ &
$\varphi(\alpha^{12} = \alpha^0\beta^4)=(1,1,0,0)=(\u_0)_4$\\

$\varphi(\alpha^5 = \alpha^2\beta^1)=(0,1,0,1)=(\u_2)_1$ &
$\varphi(\alpha^{13} = \alpha^1\beta^4)=(1,0,1,1)=(\u_1)_4$\\

$\varphi(\alpha^6 = \alpha^0\beta^2)=(0,0,1,1)=(\u_0)_2$ &
$\varphi(\alpha^{14} = \alpha^2\beta^4)=(1,0,0,1)=(\u_2)_4$\\

$\varphi(\alpha^7 = \alpha^1\beta^2)=(1,1,1,0)=(\u_1)_2$ & \\
\end{tabular}.
\end{center}
The respective Zech's logarithms are 
\[
\{\tau(i):  i \in \bbra{1,14}\}=\{4, 8, 14, 1, 10, 13, 9, 2,7,5,12,11,6,3\}.
\]
All conjugate pairs between any two nonzero cycles can be determined 
by Theorem~\ref{thm:equivalence}. Knowing $\tau(3)=14$, for example, one concludes 
that $[\u_0]$ and $[\u_2]$ share the pair $\varphi(\alpha^3)=(0,0,0,1)$ and $\varphi(\alpha^{14})=(1,0,0,1)$. Conversely, knowing a conjugate pair is sufficient to determine the corresponding Zech's logarithm. 
Since $(0,0,1,1)=\varphi(\alpha^6)$ and $(1,0,1,1)=\varphi(\alpha^{13})$ form a conjugate pair between $[\u_0]$ and $[\u_1]$, for instance, one gets $\tau(6)=13$.

The feedback function derived from $f(x)$ is $h=x_0+x_1+x_2+x_3$. By Equation (\ref{eq:CJ_feed}), the merged cycle joining $[\u_0]$ and $[\u_2]$ based on the conjugate pair $(0,0,0,1)$ and $(1,0,0,1)$ has the feedback function
$x_1 \cdot x_2 \cdot x_3 + x_1 \cdot x_3 + x_2 \cdot x_3 + x_3$. Appending $[\u_1]$ to the resulting cycle using the conjugate pair $(0,0,1,1)$ and $(1,0,1,1)$ shared by $[\u_0]$ and $[\u_1]$, results in the feedback function $\widetilde{h}=x_0 + x_1 + x_2 + x_1 \cdot x_3$ that produces a sequence of period $15$. Appending a $0$ to the string of $3$ zeroes in the sequence gives us the de Bruijn sequence $(0000~1010~0111~1011)$. \qed
\end{example}

\begin{remark}\label{rem:3}
If $f(x)$ in Theorem~\ref{thm:equivalence} is primitive, then the output is an $m$-sequence $\m=(m_0,m_1,\ldots,m_{2^n-2})$. Let $\m_0:=(m_0,\ldots,m_{n-1})=(1,\0)$. 
The $i$-th state $\m_i$ and the $\tau(i)$-th state $\m_{\tau(i)}$ form a conjugate pair. To compute $\tau(i)$ for $i \in \bbra{1,2^n-2}$ it suffices to determine the position of the state $\m_{(\tau(i))}=\m_i+\m_0$ by searching. This fact can be used to find the Zech's logarithms when $n$ is not very large. \qed
\end{remark}

\section{Spanning Trees}\label{sec:spantree}

\subsection{Constructing Enough Number of Spanning Trees}

When $n$ or a valid $t$ is large, building $\mG$ completely is possible but very often unnecessary and consumes too much resources. It is desirable to find just enough Zech's logarithms to identify a required number of spanning trees. Since there is a unique pair that joins $[\0]$ and $[\u_0]$ into one cycle, our focus is on the set of nonzero cycles $\{[\u_i] : i \in \bbra{0,t-1}\}$. 

Let $j=a_1 \cdot t +a_2$. If $\tau_n(j)= b_1 \cdot t +b_2$ with $a_2,b_2 \in \bbra{0,t-1}$, 
then $[\u_{a_2}]$ and $[\u_{b_2}]$ are adjacent. They are joined into one cycle 
by the conjugate pair
\begin{equation}\label{eq:abpair}
\v=\varphi(\alpha^j)=T^{a_1}\varphi(\alpha^{a_2}) \mbox{ and } \overline{\v}=\varphi(\alpha^{\tau_n(j)})=
T^{b_1}\varphi(\alpha^{b_2}),
\end{equation}
with $\varphi(\alpha^{a_2})$ the initial state of $\u_{a_2}$ and $\varphi(\alpha^{b_2})$ that of $\u_{b_2}$. 
We continue the process by identifying some conjugate pair(s) between enough pairs of adjacent cycles until all of the cycles in $\Omega(f(x))$ can be joined. 

Let $\tau(j)$ for some $j \in \bbra{1,2^n-2}$ be known. This induces a mapping from $D_j$ onto 
$D_{\tau(j)}$ with $n_j:=|D_j|=|D_{\tau(j)}|$. If $j \not \equiv \tau(j) \Mod{t}$, then, 
for $i \in \bbra{0,n_j-1}$, states $\varphi(\alpha^{2^ij})$ and $\varphi(\alpha^{2^i\tau(j)})$ belong to 
distinct cycles. These states join their corresponding cycles into one. Let $m_j$ be the least positive integer such that $(2^{m_j}-1) j \equiv 0 \Mod{t}$. 
Observe that $\varphi(\alpha^{j})$ and $\varphi(\alpha^{j \cdot 2^{m_j}})$
are states of the same cycle and $m_j \mid n_j$. Hence, given cosets $D_j$ 
and $D_{\tau(j)}$, one derives $\frac{n_j}{m_j}$ 
distinct conjugate pairs between each of the $m_j$ distinct pairs of cycles.

The Zech's logarithms supply the exact positions of the conjugate pair(s) in the relevant 
cycles. Once enough conjugate pairs to construct a spanning tree are identified, 
the precise positions to apply the cycle joining method, \ie, to exchange successors, appear. 
Thus, with $(1,\0)$ as the initial state of $\u_0$, we just need to keep 
track of the precise positions, in terms of the operator $T$ and the power of $\alpha$, 
governed by the $(j,\tau_n(j))$ pair. The actual construction of the de Bruijn sequences 
no longer requires storing the initial states of the $t$ nonzero sequences 
in $\Omega(f(x))$. 

\begin{example}\label{ex:300}
Consider $p(x)=x^{300}+x^7+1$ and let $\alpha$ be a root of $p(x)$, implying $\tau(7)=300$. 
Choosing $t =31$, the Berlekamp-Massey algorithm outputs 
\[
f(x)= x^{300}+x^{194}+x^{176}+x^{158}+x^{97}+x^{88}+x^{79}+ 	  x^{52}+x^{43}+x^{25}+x^{16}+x^7+1.
\]
Hence, $\Omega(f(x))=[\0]~\cup~\bigcup_{i=0}^{30} [\u_i]$. 
Let $(1,\0)\in\F_2^{300}$ be the initial state of $\u_0$. 
Knowing a specific $(i,\tau(i))$ pair gives us $\{(j,\tau(j)): j \in D_{i}\}$. 
Note that $|D_{i}|=300$ for all $i$.

Since $\tau(7) \equiv 21 \Mod{31}$, there are $5$ distinct pairs of cycles, each sharing $60$ conjugate pairs. 
The $(a_2,b_2)$ pairs $(7,21), (14,11), (28,22), (25,13), (19,26)$ are the indices of the cycles. 
One of the $60$ conjugate pairs between $[\u_{7}]$ and $[\u_{21}]$ is 
$(\varphi(\alpha^{7}), T^{9} \varphi(\alpha^{21}))$ since $\floor{\frac{7}{31}}=0$ and $\floor{\frac{300}{31}}=9$. 
Computing $a_1$ and $b_1$ are easy given the relevant logarithms, so we omit them from the rest of this example. 

Since $\tau(1) \equiv \tau(3) \equiv 0 \Mod{31}$, $[\u_0]$ shares $60$ conjugate pairs each with $[\u_j]$ for $j \in \{1,2,4,8,16\} \cup \{3,6,12,24,17\}$. Similarly, each adjacent cycles in the list

\begin{center}
\begin{tabular}{cl }
$(i, \tau(i) \pmod{31})$ & Indices of Adjacent Cycles \\
\hline
$(5,3)$ & $(3,5),(6,10),(12,20),(24,9),(17,18)$ \\

$(15,22)$ & $(21,27),(11,23),(22,15),(13,30),(26,29)$\\

$(35,7)$ & $(4,7),(8,14),(16,28),(1,25),(2,19)$\\
\end{tabular}
\end{center}
shares $60$ conjugate pairs. We order the cycles as $[\0],[\u_0],[\u_1],\ldots,[\u_{30}]$ and build an adjacency subgraph $\widetilde{\mG}$ from the  
computational results. Applying Theorem~\ref{BEST} with 
$\mG$ replaced by $\widetilde{\mG}$, the approach produces $\approx 2^{177.21}$ de Bruijn sequences. Figure~\ref{fig:graph} is a spanning tree. \qed

\begin{figure}[h]
\begin{tikzpicture}[auto, node distance=1.1cm, every loop/.style={},
                    thick,main node/.style={draw,font=\sffamily\bfseries}]

  \node[main node] (1) {$[\u_0]$};
  \node[main node] (2) [left of=1] {$[\0]$};
  \node[main node] (3) [above of=1] {$[\u_4]$};
  \node[main node] (4) [left of=3] {$[\u_2]$};
  \node[main node] (5) [left of=4] {$[\u_1]$};
  \node[main node] (6) [right of=3] {$[\u_8]$};

  \node[main node] (7) [right of=6]{$[\u_{16}]$};
  \node[main node] (8) [above of=3] {$[\u_{7}]$};
  \node[main node] (9) [left of=8] {$[\u_{19}]$};
  \node[main node] (10) [left of=9] {$[\u_{25}]$};
  \node[main node] (11) [right of=8] {$[\u_{14}]$};
  \node[main node] (12) [right of=11] {$[\u_{28}]$};
  
  \node[main node] (13) [above of=8] {$[\u_{21}]$};
  \node[main node] (14) [left of=13] {$[\u_{26}]$};
  \node[main node] (15) [left of=14] {$[\u_{13}]$};
  \node[main node] (16) [right of=13] {$[\u_{11}]$};
  \node[main node] (17) [right of=16] {$[\u_{22}]$};
  
  \node[main node] (18) [above of=13] {$[\u_{27}]$};
  \node[main node] (19) [left of=18] {$[\u_{29}]$};
  \node[main node] (20) [left of=19] {$[\u_{30}]$};
  \node[main node] (21) [right of=18] {$[\u_{23}]$};
  \node[main node] (22) [right of=21] {$[\u_{15}]$};

  \node[main node] (23) [below of=1] {$[\u_{12}]$};
  \node[main node] (24) [left of=23] {$[\u_{6}]$};
  \node[main node] (25) [left of=24] {$[\u_{3}]$};
  \node[main node] (26) [right of=23] {$[\u_{24}]$};
  \node[main node] (27) [right of=26] {$[\u_{17}]$};
  
  \node[main node] (28) [below of=23] {$[\u_{20}]$};
  \node[main node] (29) [left of=28] {$[\u_{10}]$};
  \node[main node] (30) [left of=29] {$[\u_{5}]$};
  \node[main node] (31) [right of=28] {$[\u_{9}]$};
  \node[main node] (32) [right of=31] {$[\u_{18}]$};
    
  \path[every node/.style={font=\sffamily}]
    (1) edge (2) 
        edge (3)
        edge (4) 
        edge (5)
        edge (6) 
        edge (7)
        edge (23) 
        edge (24)
        edge (25) 
        edge (26)
        edge (27) 
    (5) edge (10)
    (10) edge (15) 
	(15) edge (20) 
	
	(4) edge (9)
    (9) edge (14) 
	(14) edge (19)
	
	(3) edge (8)
    (8) edge (13) 
	(13) edge (18)
	
	(6) edge (11)
    (11) edge (16) 
	(16) edge (21)
	
	(7) edge (12)
    (12) edge (17) 
	(17) edge (22)
	
	(25) edge (30)
    (24) edge (29)
	(23) edge (28)
	
	(26) edge (31)
    (27) edge (32);

\end{tikzpicture}
\centering
\caption{A spanning tree in an adjacency subgraph $\widetilde{\mG}$ of $\Omega(f(x))$.}
\label{fig:graph}
\end{figure}
\end{example}

\subsection{A Note on Dong and Pei's Construction}

Dong and Pei recently proposed a construction of de Bruijn sequences with large order in~\cite{Dong16}. Given an irreducible polynomial $f(x)$ of degree $n$, order $e$, and $t=\frac{2^n-1}{e}$, they defined a shift register matrix $T$ in the form of (\ref{comx}) satisfying $f^*(T)=0$ where $f^*(x)$ is the reciprocal polynomial of $f(x)$. 
Given the sequence $\u_0$ with initial state $\alpha_{0}=(1,\0)$, one can write 
any sequence as $g(T) \u_0$, where $g(x)$ is some polynomial of degree $< n$. 
If $(1+x^k)^e \not \equiv 1 \Mod{f(x)}$, then $[\u_0]$ and $[(1+T^k)^{2^j}\u_0]$ are distinct cycles sharing the conjugate pair $\left(T^{k 2^j}\alpha_{0}, (1+T^{k})^{2^j} \alpha_{0} \right)$. Here $T^{k 2^j}\alpha_{0}$ is a state of $[\u_0]$. Their claim that $[\u_0]$ shares some conjugate pairs with each of the other nonzero cycles does {\bf not} hold in general. 

First, as $n$ and $e$ grow large, computing $(1+x^k)^e \Mod{f(x)}$ soon becomes prohibitive. 
Second, after $(1+x^k)^e \not \equiv 1\Mod{f(x)}$ is verified, it remains unclear which cycle $[(1+T^k)^{2^j} \u_0]$ corresponds to. One is left unable to judge whether it is possible to join all of the cycles in $\Omega(f(x))$ 
even after a lot of the conjugate pairs have been determined. 
Third, and most importantly, $t < \sqrt{2^n-1}$ is a necessary condition for 
their method to work~\cite[Section~5]{Dong16}. In fact, a sufficient and necessary condition is $(0,i)_t > 0$ for all $i \in \bbra{1,t-1}$. This does not hold in general. Take, for example, $n=10$ with $p(x)=x^{10}+x^3+1$ and $t=31 < \sqrt{31 \cdot 33}$. All values $1 \leq i \leq 2^{10}-2$ 
such that $\tau(i) \equiv 0 \Mod{t}$ form the set  
\begin{align*}
X:=\{&85,105,141,170,210,277,282,291,325,337,340,341,379,420,431,493,\\
     &554,564,582,650,657,674,680,682,701,727,758,840,862,875,949,986\}.
\end{align*}
Hence, $[\u_0]$ can be joined only to $[\u_{\ell}]$ with 
\[
\ell \in \{3, 6, 7, 12, 14, 15, 17, 19, 23, 24, 25, 27, 28, 29, 30 \}.
\]
Since only $15$ out of the required $30$ cycles can be joined with $[\u_0]$, 
Dong and Pei's approach fails to produce de Bruijn sequences here. We show in the next subsection that our method handles such a situation perfectly.

\subsection{Star and Almost-Star Spanning Trees}

In cases where $[\u_0]$ is adjacent to $[\u_j]$ for $j \in \bbra{1,t-1}$ we can use Theorem~\ref{thm:equivalence} to rapidly certify the existence of star spanning trees centered at $[\u_0]$ in $\mG$. The main idea is to build $\widetilde{\mG}$ such that it becomes a star graph centered at $[\u_0]$ with leaves $[\0]$ and $[\u_{j}]$ for $j \in \bbra{1,t-1}$ when we replace all of the edges between $[\u_0]$ and $[\u_j]$ by a single edge. The certificate contains the following outputs of Algorithm~\ref{algo:cert}.
\begin{enumerate}
\item A {\it witness} $\mW$ that generates $\Delta_{\mW} \triangleq 
\{i:=k \cdot t \mbox{ for } k \in \mW\}$ satisfying
\begin{equation}\label{eq:witness}
\bbra{1,t-1} \subset \bigcup_{i \in \Delta_{\mW}} 
\{j \Mod{t} : j \in D_{Y_{i}}\} \mbox{ with } Y_i:=\tau(i) \Mod{t}.
\end{equation}

\item The number $\#cp$ of identified conjugate pairs between $[\u_0]$ and $[\u_j]$.
\item A matrix $\widetilde{M}$ derived from the adjacency subgraph $\widetilde{\mG}$.
\item The exact number $\#dBSeqs$ of de Bruijn sequences that can be constructed. 
\end{enumerate} 

\begin{algorithm}[h!]
\caption{Certifying the Existence of a Star Spanning Tree}
\label{algo:cert}
\begin{algorithmic}[1]
\renewcommand{\algorithmicrequire}{\textbf{Input:}}
\renewcommand{\algorithmicensure}{\textbf{Output:}}
\Require A primitive polynomial $p(x)$ of degree $n$.
\Ensure Witness $\mW$, $\#cp$, Matrix $\widetilde{M}$ and $\#dBSeqs$.
\State{$N \gets 2^n-1$}
\State{$\F_{2^n} := \F_2(\alpha) \gets$ the extension field of $\F_2$ defined by $p(x)$}
\Comment{$\alpha$ is a root of $p(x)$}
\State{${\rm Div} \gets \{d \mbox{ such that } d \mid 2^n-1 \mbox{ and the minimal polynomial of } \alpha^d \mbox{ is of degree } n\}$}
\For{$t \in {\rm Div}$} \label{line:s}
	\State{$f(x) \gets \mbox{ the minimal polynomial of } \beta:=\alpha^t$} 
	\State{Initiate sets ${\tt Done} \gets \{0\}$, ${\tt MinSet} \gets \emptyset$, and $\mW \gets \emptyset$}\label{line9}
	\State{$\widetilde{M}=(m_{i,j}) \gets t \times t$ zero matrix}
	\For{$i \in \{(2k-1) t : 1 \leq k \leq z\}$} \label{line:z} \Comment{Letting $z=2000$ suffices for $n \leq 100$}
		\State{$L \gets \tau(i) \Mod{t}$}
		\State{$c_{L} \gets$ coset leader of $D_{L}$}
		\If{$c_{L} \Mod{t} \notin {\tt Done}$}
			\State{${\tt MinSet} \gets \{i\} \cup {\tt MinSet}$}
			\State{$\mW \gets \{2k-1\} \cup \mW$}
			\State{${\tt Done} \gets {\tt Done} \cup \{y \Mod{t}: y \in D_{L}\}$}
			\If{$\abs{{\tt Done}} = t$}
				\State{$\#cp \gets \frac{n}{\abs{D_{L}}}$}
				\State{$m_{1,1} \gets \#cp \cdot (t-1)+1$}\label{line20}
				\For{$r$ from $2$ to $t$}
					\State{$m_{r,r} \gets \#cp$}
					\State{$m_{1,r}=m_{r,1} \gets -\#cp$}
				\EndFor\label{line25}
				\State{output $\mW$ and $\#dBSeqs \gets \det(\widetilde{M})$}
				\State{break $i$}
			\EndIf
		\EndIf
	\EndFor
\EndFor	
\end{algorithmic}
\end{algorithm}

If no star spanning tree can be certified, one can instead starts by finding an index $\ell \in \bbra{1,t-1}$ such that $[\u_{\ell}]$ is the center of a star graph with $[\u_j]$ for all $j \in \bbra{0,t-1} \setminus \{\ell\}$ as its leaves. We extend this star graph by appending vertex $[\0]$ and the unique edge $E_0$ between $[\0]$ and $[\u_0]$ into a spanning tree in $\mG$. We build $\widetilde{\mG}$ by collecting such trees. For convenience we call them {\it almost-star spanning trees centered at $[\u_{\ell}]$} in $\mG$. 
 
A modified version of Algorithm~\ref{algo:cert} can be used to certify the existence of almost-star spanning trees. We replace $0$ by $\ell$ in Line~\ref{line9} and replace $(2k-1) t$ in Line~\ref{line:z} by $(2k-1) t + \ell$. The entries of $\widetilde{M}$ defined in Lines~\ref{line20} to~\ref{line25} are now given as follow.

\begin{algorithmic}[1]
\State{$m_{1,1} \gets 1+\#cp$, $m_{1,\ell+1}=m_{\ell+1,1} \gets -\#cp$}
\For{$r$ from $2$ to $t$}
\State{$m_{r,r} \gets \#cp$, $m_{\ell+1,r}=m_{r,\ell+1} \gets -\#cp$}
\EndFor
\State{$m_{\ell+1,\ell+1} \gets \#cp \cdot (t-1)$}
\end{algorithmic}

\begin{example}\label{ex:star}
There is no star spanning tree centered at $[\u_0]$ for $p(x)=x^{10}+x^{3}+1$ and $t=31$. There are $\approx 2^{99.66}$ almost-star spanning trees centered at $[\u_6]$ with witness $\mW=\{1,3,7,9,13,17,21\}$ and $f(x)=x^{10}+x^9+x^5+x+1$. 

If $n=20$, then there are $38$ choices for $t$, with $23$ of them being less than $1000$. There is no star spanning tree certificate for $p(x)=x^{20}+x^3+1$ and $t \in K:=\{165,341,451,465,615,775,825\}$.  Going through $\ell \in \bbra{2,t-1}$ produces certificates for almost-star spanning trees for all $t \in K$. Based on the $\ell$ values that yield the most number of de Bruijn sequences, we choose $\ell=2$ for $t \in \{165,341,451,465\}$, $\ell=4$ for $t=775$, and $\ell=10$ for $t \in \{615,825\}$. Among these combinations of $t$ and $\ell$, the smallest number of de Bruijn sequences produced is from $t=165$ and $\ell=2$, which is $\approx 2^{708.80}$. The largest, which is $\approx 2^{3345.17}$, comes from $t=775$ and $\ell=4$. \qed
\end{example}

There are parameter sets for which there is a unique star spanning tree, yielding only $1$ de Bruijn sequence. Examples include 
\begin{align*}
p(x) & =x^{20}+x^3+1 \mbox{ with } t \in \{41,123,205,275\},\\ 
p(x) & =x^{29}+x^2+1 \mbox{ with } t \in \{233,1103,2089\}, \mbox{ and}\\ 
p(x) & =x^{130}+x^3+1 \mbox{ with } t=131.
\end{align*}
There are certificates for almost-star 
spanning trees for all of them, ensuring the existence of a large number of de Bruijn sequences in each case.

One can of course utilize both types of certificates from the same input parameters. Counting the number of, respectively, star and almost-star spanning trees with $\ell=2$, for $p(x)=x^{128}+x^{7}+x^{2}+x+1$ and $t=255$ gives us $2^{1524}$ and $2^{1778}$ de Bruijn sequences while~\cite[Example 3]{Dong16} yields $2^{1032}$ sequences.

\begin{sidewaystable}[p]
\caption{Examples of Star and Almost-Star Spanning Tree Certificates. For Star, the center is $[\u_0]$ while for Almost-Star the center is $[\u_{\ell}]$.}
\label{table:star}
\renewcommand{\arraystretch}{1.3}
\centering
\begin{tabularx}{\textwidth}{ccccl|lccl}
\hline
No. & $n$ & $p(x)$ & $t$ & $f(x)$ & Star Witness $\mW$ & $\#cp$ & $\#dBSeqs$ & Time \\
\hline
$1$ & $100$ & $\{37\}$ & $25$ & $\{96,68,64,37,36,32,4\}$ & $\{1,15\}$ & $25$ & $\approx 2^{111.45}$ & $0.14$s \\
		
$2$ & $128$ & $\{7,2,1\}$ & $255$ & $\{128,127,126,125,123,122,119,118,117,114,112,$ & $\{1,3,5,7,9,11,13,15,17,$ & $64$ & $2^{1524}$ & $1$m$13$s\\ 
& &  & & $~~109,108,106,105,103,99,98,96,94,93,91,89,87,$ & $~~19,21,27,29,31,33,37,$ & &  & \\
& &  & & $~~83,82,81,74,70,68,67,66,65,63,60,59,56,53,$ & $~~43,45,47,53,57,65,77,$ & &  & \\
& & & & $~~52,51,50,46,45,44,43,42,41,40,39,38,36,31,$ & $~~79,83,101,107,123,133,$  & & & \\
& & & & $~~30,29,28,27,26,25,24,22,15,13,11,10,9,6,4,1\}$ & $~~141,145,177,187,929\}$  & & & \\
		
$3$ & $130$  & $\{3\}$  & $93$ & $\{97,89,64,63,48,47,43,42,21,11,10,5,3,2,1\}$ & $\{1,3,5,9,11,15,23,31,33,$ & $26$ & $\approx 2^{432.44}$ & $0.78$s \\
&&&&& $~~35,43,73,101\}$ && \\
		
$4$ & $300$ & $\{7\}$ & $77$ & $\{273,220,219,193,192,191,165,164,139,111,$ & $\{1,3,5,97,125\}$ & $100$ & $\approx 2^{504.93}$ & $5$m$41$s \\
& & & & $~~110,86,85,83,82,30,29,28,7,6,5,4,3,2,1\}$ & & & & \\
		
\hline
No. & $n$ & $p(x)$ & $t$ & $f(x)$ & Almost-Star Witness $\mW$ and $\ell$ & $\#cp$ & $\#dBSeqs$ & Time \\
\hline
		
$5$ & $20$ & $\{3\}$ & $205$ & $\{18,17,15,14,9,8,4,2,1\}$ & $\{1,3,5,7,9,11,21,23,25,$ & $20$ & $\approx 2^{881.67}$ & $0.02$s \\
& & & & & $~~41,53,155\}$, $\ell=2$  & & & \\
		
$6$ & $29$ & $\{2\}$ & $233$ & $\{24, 22,20,18,16,15,14,13,12,11,$ & $\{1,3,5,9,15,17,19,33,79\}$, $\ell=2$ & $29$ & $\approx 2^{1127.05}$ & $0.08$s \\
& & & & $~10,9,8,7,5,4,2\}$ & & & & \\
		
$7$ & $128$ & $\{7,2,1\}$ & $255$ & See Entry $2$ above & $\{1,3,5,7,9,11,13,17,19,23,25,$ & $128$ & $2^{1778}$ & $1$m$24$s\\ 
& &  & &  & $~~29,31,33,37,45,49,53,55,57,$ & &  & \\
& &  & &  & $~~67,89,91,103,107,111,119,$ & &  & \\
& & & &  & $~~139,143,159,201,237,251,$  & & & \\
& & & &  & $~~343,465\}$, $\ell=2$ & & & \\
		
$8$ &  $130$ & $\{3\}$ & $131$ & $\{96,72,65,48,36,34,24,17,12,10,5,4,2\}$ & $\{1,171\}$, $\ell=2$ & $130$ & $\approx 2^{912.91}$ & $0.86$s \\
\hline
\end{tabularx}
\end{sidewaystable}

Table~\ref{table:star} lists more examples. For a compact presentation we use sparse primitive polynomials. They are either trinomials, \ie, $p(x)=x^n+x^k+1$ with $1 \leq k <n$, or $p(x)=x^n+x^k+x^j+x^i+1$ with $1 \leq i < j < k < n$. 
Given $n$ and $t$, the  chosen polynomial $p(x)$ and its corresponding $f(x)$ are presented as sets whose elements are the powers of $x$, from $1$ to $n-1$, whose coefficients are $1$. Hence, for $n=130$ and $t=31$,  
$p(x)=x^{130}+x^{3}+1$ and $f(x)=x^{130}+ x^{63} + x^{31}+x^{15} + x^{7} + x^{3}+1$. 
Its witness $\mW=\{1,3,7,9,17,45\}$ builds $\Delta_{\mW}=\{31,93,217,279,527,1395 \}$, 
implying $\{Y_i=\tau(i) \Mod{31} : i \in \Delta_{\mW}\}=\{20,16,14,13,23,12\}$. The corresponding sets $\{j \Mod{t} : j \in D_{Y_{i}}\}$ are
\begin{align*}
&\{20,9,18,5,10\}, \{16,1,2,4,8\}, \{14,28,25,19,7\},\\
&\{13,26,21,11,22\},\{23,15,30,29,27\},\{12,24,17,3,6\}.
\end{align*}
Their union is $\bbra{1,30}$. Note that $[\u_0]$ and $[\u_{\ell}]$ share 
$\frac{130}{5}=26$ conjugate pairs for $\ell \in \bbra{1,30}$. Computing for $\#dBSeqs$ is then straightforward. The other entries can be similarly interpreted. 
The recorded running time is for the specified $(n,p(x),t)$ with $\ell$ added for cases where the center of the almost-star trees is 
$[\u_{\ell}]$. 

The certificates are easily computable, assuming commonly available resources in modern PCs. We performed the computations for all of the examples above on a laptop with Ubuntu 16.04 OS powered by an Intel i7-7500U CPU 2.90GHz, running MAGMA V2.24-5 \cite{BCP97} on $11.6$ GB of available memory. In fact, we computed the certificates on input $p(x)=x^{128} + x^7 + x^2 + x + 1$ and $t \in \{1,3,5,15,17,51,85\}$ on the online calculator \url{http://magma.maths.usyd.edu.au/calc/}. Each $t$ takes less than the $120$s time restriction when $z$ is set to $100$.

To get a better sense of the complexities, we implemented Algorithm~\ref{algo:cert} and its modification on all $4 \leq n \leq 100$ and $n \in \{113,119,128,130,256,300,512\}$. We use the Finite Field Logarithm Database ({\tt libs/data/FldFinLog}) of MAGMA to handle $n > 130$. Except for a number of cases, it suffices to bound $ t \leq 100$ on Line~\ref{line:s} and set $z$ in Line~\ref{line:z} to $2000$. Doing so already guarantees us a very large number of spanning trees. This keeps the running time and memory demand low. On the already described machine, with the added restrictions, the entire simulation for all $n$ and $t$ combinations took less than $50$ hours to complete. Without loss of generality, due to Proposition~\ref{prop:commdiag}, we use the default choice of $p(x)$ in MAGMA. Interested readers may contact the corresponding author for the simulation data.

What about the exceptional cases? For $n \in \{5,7,13,17,19,31,61,89,107,127\}$ there is simply no valid $t > 1$. Table~\ref{table:splist} lists $n \leq 128$ and its smallest $t>1$. The memory demand for the shaded entries exceeds $16$~GB. The rest of the entries can be easily handled to obtain the relevant star and almost-star certificates.

\begin{table}[h!]
\caption{Exceptional $n$ for which the smallest $t > 1$ is larger than $100$}\label{table:splist}
\renewcommand{\arraystretch}{1.1}
\centering
\begin{tabular}{c l | c l | c l | c l}
\hline
$n$ & $t$ & $n$ & $t$ & 
$n$ & $t$ & $n$ & $t$ \\
\hline
$29$ & $233$   & $49$ & $127$ &  $73$ &  $439$ & $\colorbox[gray]{0.8}{101}$ & $\colorbox[gray]{0.8}{7432339208719}$ \\
$37$ & $223$   & $53$ & $6361$ & $79$ &  $2687$ & $\colorbox[gray]{0.8}{103}$ & $\colorbox[gray]{0.8}{2550183799}$\\
$41$ & $13367$ & $\colorbox[gray]{0.8}{59}$ & $\colorbox[gray]{0.8}{179951}$ 
& $83$ & $167$ & $\colorbox[gray]{0.8}{109}$ & $\colorbox[gray]{0.8}{745988807}$\\
$43$ & $431$   & $\colorbox[gray]{0.8}{67}$ & $\colorbox[gray]{0.8}{193707721}$ & $91$ & $127$ & $113$ & $3391$ \\
$47$ & $2351$  & $\colorbox[gray]{0.8}{71}$ & $\colorbox[gray]{0.8}{228479}$ & $97$ & $11447$ & $119$ & $127$ \\
\hline
\end{tabular}
\end{table}

A star spanning tree in $\mG$ exists if and only if the cyclotomic number $(0,j)^n_t > 0$ for the given $n$ and $t$ and for all $j \in \bbra{1,t-1}$. Similarly, an almost-star spanning tree centered at $[\u_{\ell}]$ exists if and only if the cyclotomic number $(\ell,j)^n_t > 0$ for the given $n$ and $t$ and for all $j \in \bbra{0,t-1} \setminus \{\ell\}$. There are a few cases for which neither a star nor an almost-star certificate is found. Some of the examples are $n=10$ with $t=93$, $n=11$ with $t=89$, and $n=12$ with $t \in \{91,105,117,315\}$.

\section{A Basic Software Implementation for Cycle Joining}\label{sec:program}

We make a basic implementation of the cycle joining method via Zech's logarithms  available online in~\cite{EFZech}. It requires {\tt python 2.7} equipped with {\tt sympy} and some access to MAGMA, either locally or through internet connection to its online calculator. Many users care more about the speed of generation and need only a relatively small numbers of de Bruijn sequences. Others may prefer completeness and would like to produce all de Bruijn sequences that can be generated from the given input polynomial and a stipulated number $t$ of cycles. We cater to both types of users. 

Given an order $n$, the software takes as input a primitive polynomial $p(x)$ of order $n$ supplied by a user or let MAGMA choose such a $p(x)$. The set ${\rm Div}$ of all valid $t$ is computed. When one such $t$ is chosen by the user, the routine generates the corresponding irreducible polynomial $f(x)$ and a complete set of coset representatives $\mathcal{R}_n$. It seems that there is no closed formula for $\abs{\mathcal{R}_n}$, which is the number of monic irreducible factors of $x^{2^n-1}-1$ over $\F_2$. Based on $f(x)$, the cycle structure is now completely determined once the program ensures that $(1,\0)$ is the initial state of $[\u_0]$. Depending on the user's requirements, there are several options on offer. 

\begin{enumerate}
	\item For small orders, say $n \leq 30$, the program stores the Zech's logarithms of the elements in $\mathcal{R}_n$ in a dictionary form. The filesize is in tens of MB for $n=30$ and tends to double from $n$ to $n+1$. The complete adjacency graph $\mG$ can then be generated by using the {\tt Double} map.
	\item Given combinations of $n$ and $t$ for which it is practical to produce witnesses for enough number of star or almost-star spanning trees, the program finds the Zech's logarithms of the elements in the witness set $\mathcal{W}$ and stores the corresponding adjacency subgraph $\widetilde{\mG}$.
	\item The program asks MAGMA for the Zech's logarithms of a few random entries, based on $p(x)$, and applies the maps in Subsection~\ref{subsec:zech} to generate a partial Zech's logarithm table. It attempts to build $\widetilde{\mG}$ from the table. If not all vertices are connected, it identifies pairs of vertices that still need to be connected and uses the information to get the required Zech's logarithms. It stops when a connected adjacency subgraph containing all vertices is obtained. 
\end{enumerate}

Finally, it chooses the required number of spanning trees in $\widetilde{\mG}$ or, in Option 1, those in $\mG$, and produces de Bruijn sequences using the procedure explained in~\cite[Section VI]{Chang2019}. Various modifications can be done on the routine to make it run faster, \eg, if generating only a few de Bruijn sequences is all a user wants. 

\begin{example}
There are $\approx 2^{145.73}$ de Bruijn sequences that can be produced on input $p(x)=x^{10}+x^3+1$ and $t=31$. It takes $0.072$ seconds and a negligible amount of memory to output one of them. The simplified adjacency subgraph, where multiple edges between two vertices is presented as one, is in Figure~\ref{fig:10} with vertex $i \in \bbra{1,31}$ representing $[\u_i]$ and $0$ representing $[\0]$. On input $p(x)=x^{22}+x+1$ and $t=89$, it takes $5$ minutes $30$ seconds and $680$ MB of memory to print one of the $\approx 2^{425.42}$ resulting sequences out to a file. 

The program comes with an option to produce the ANF, instead of the actual sequence. Here is such an instance on input $p(x)=x^6 + x + 1$ and $t=7$. We have $f(x)= x^6 + x^3 + 1$. One of the resulting de Bruijn sequences is 
\[
(00000011~01110011~10110001~11111000~01011110~10101101~00010010~10011001),
\]
whose ANF is 
\begin{multline*}
x_0 + x_1   x_2   x_3   x_4   x_5 + x_1   x_2   x_3   x_5 + x_1   x_2   x_4   x_5 + x_1   x_2   x_5 + x_1   x_3   x_4 \\
+ x_1   x_3 + x_1   x_4 + x_1 + x_2   x_3 + x_2 + x_3   x_4   x_5 + x_4   x_5 + 1.
\end{multline*}
\qed

\begin{figure}[!t]
\centering
\includegraphics[height=80.0mm,keepaspectratio]{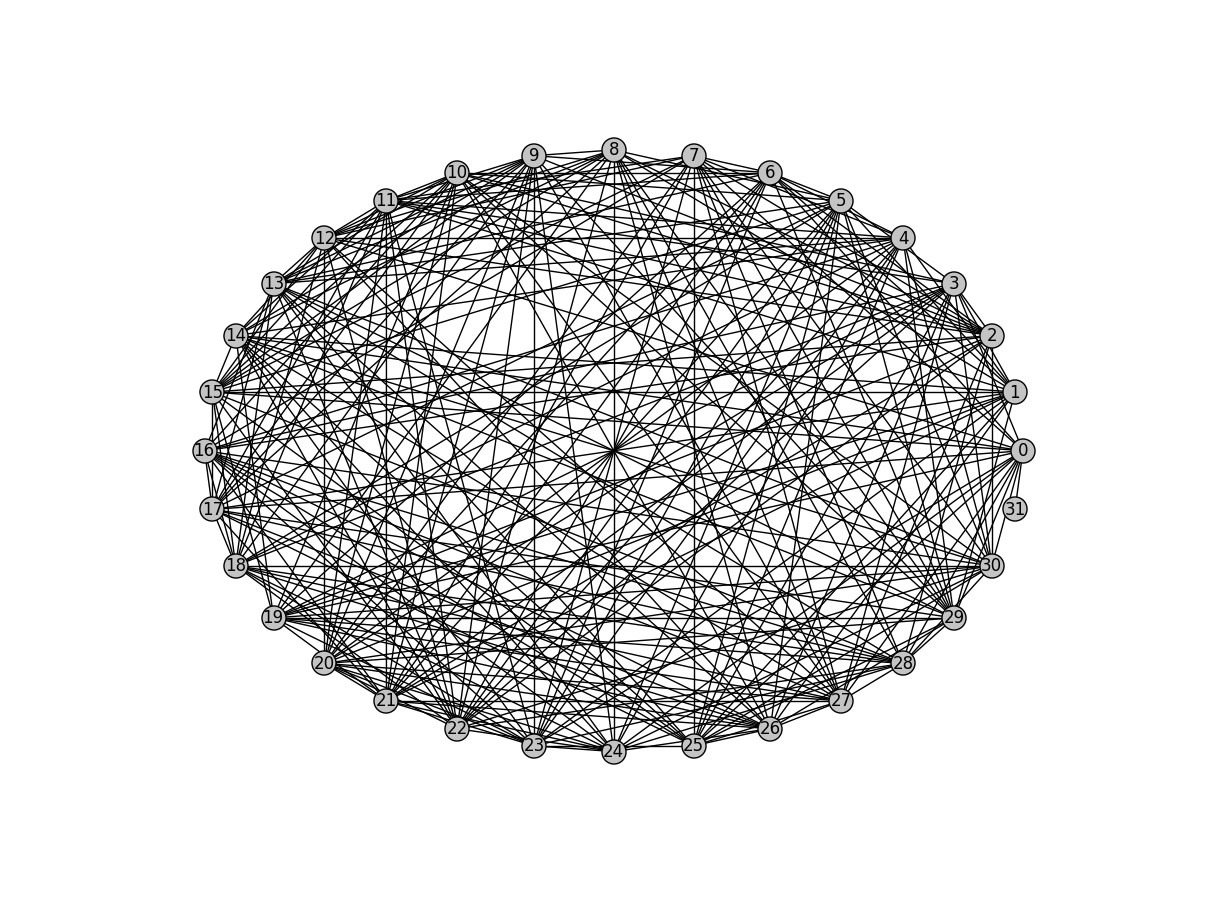}
\caption{A simplified adjacency subgraph $\widetilde{\mG}$ for $f(x)=x^{10}+x^3+1$ with $t=31$.}
\label{fig:10}
\end{figure}

\end{example}

\noindent{\bf Harvesting the cyclotomic numbers} \\
So far there is no closed formula for cyclotomic numbers defined in (\ref{eq:cycnum}), except in several limited classes. Storer's textbook~\cite{Storer67} and a more recent one by Ding~\cite{Ding14} have more details. Determining the values computationally is a by-product of our implementation in which we obtain the respective cyclotomic numbers on the input parameters specified in Table~\ref{table:cycnumbers}. They are available for download in our online repository~\cite{EFZech}. 

\begin{table}[h]\caption{The Input Parameters of the Computed $(i,j)_t^n$ for $0 \leq i,j \leq t$}
\label{table:cycnumbers}
\renewcommand{\arraystretch}{1.1}
\centering
\begin{tabular}{cll}
\hline 
$n$  & Set of feasible $t$	&  A suitable $p(x)$  \\  
\hline 
$4$	& $\{3\}$  & $x^4+x+1$  \\

$6$  & $\{3,7\}$	&  $x^6+x+1$ \\ 

$8$ & $\{3,5,15\}$ & $x^8+x^5+x^3+x+1$ \\

$9$ & $\{7\}$	& $x^9 + x^4 + 1$ \\

$10$  & $\{3,11,31,93\}$ & $x^{10}+x^3+1$ \\

$11$  & $\{23,89\}$	&  $x^{11}+x^2+1$  \\ 

$12$  & $\{3,5,7,9,13,15,21,35,39,45,63,91,105,117,315\}$ &  $x^{12}+x^8+x^2+x+1$ \\
$14$ & $\{3,43,127,381\}$ &  $x^{14}+x^{12}+x^2+x+1$\\

$15$ &$\{7,31,151,217\}$ &  $x^{15}+x^{4}+x^2+x+1$\\

$16$ & $\{3,5,15,17,51,85,255\}$ & $x^{16}+x^{12}+x^{3}+x+1$\\

\hline
\end{tabular} 
\end{table}

\section{Product of Irreducibles}\label{sec:prod}

The approach via Zech's logarithms can be used in tandem with the one in~\cite{Chang2019} to generate de Bruijn sequences of even larger orders. To keep the exposition brief, we retain the notations from the said reference.

Let $\{f_1(x), f_2(x), \ldots,f_s(x)\}$ be a set of $s$ pairwise distinct irreducible polynomials
over $\F_2$. Each $f_i(x)$ has degree $n_i$, order $e_i$ with $t_i=\frac{2^{n_i}-1}{e_i}$, 
and a root $\beta_i$. Let the respective associated primitive polynomials be $p_i(x)$ with 
degree $n_i$ and root $\alpha_i$. Hence,
$\Omega(f_i(x))=[\0]\cup[\u_0^i]\cup[\u_1^i]\cup\ldots\cup[\u_{t_i-1}^i]$. Let the initial state of $\u_0^i$ be $(1,\0)\in\F_2^{n_i}$. The initial states 
of $\u_j$ for $j \in \bbra{1,t_i-1}$ follows by decimating the appropriate $m$-sequence 
$\m_i$ generated by $p_i(x)$. For the rest of this section, let
\begin{equation}\label{eq:prod}
f(x):=\prod_{i=1}^{s} f_i(x) \mbox{ and } n:=\sum_{i=1}^s n_i.
\end{equation}

We use the expression for the cycle structure of $\Omega(f(x))$ given in~\cite[Lemma 3 Eq.~(7)]{Chang2019}. 
For any cycle $\Gamma_1:=[\u^1_{i_1}+L^{\ell_2}\u^2_{i_2}+\cdots+L^{\ell_s}\u^s_{i_s}]$ containing 
a state $\v$ the goal is to identify a cycle $\Gamma_2$ that contains $\overline{\v}$. Letting $\mP$ be the matrix defined in~\cite[Section III.B]{Chang2019}, 
$\v=(\v_1,\ldots,\v_s) \mP$ with $\v_i:=\varphi(\alpha_i^{j_i})$ for 
$i \in \bbra{1,s}$. One then gets a state $\a_i$ of some nonzero sequence in $\Omega(f_i(x))$ satisfying $(\a_1,\ldots,\a_s)\mP=(1,\0)$. The exact position of each $\a_i$ 
in the corresponding cycle in $\Omega(f_i(x))$, \ie, the exact value of $\gamma_i$
satisfying $\a_i=\varphi(\alpha_i^{\gamma_i})$ is computed using the method from Section~\ref{sec3} or by an exhaustive search when $n_i$ is small. The conjugate state $\overline{\v}=(\b_1,\ldots,\b_s)\mP$ of $\v$ must then be 
\begin{equation}\label{eq:conjupi}
\b_i=\varphi(\alpha_i^{j_i})+\varphi(\alpha_i^{\gamma_i})=\varphi(\alpha_i^{j_i}+\alpha_i^{\gamma_i})
=\varphi(\alpha_i^{\gamma_i}(1+\alpha_i^{j_i-\gamma_i})) =\varphi(\alpha_i^{\gamma_i+\tau_i(j_i-\gamma_i)}),
\end{equation}
with $\tau_i$ based on $p_i(x)$. 
If $\b_i$ is the $j_i$-th state of $\u^i_{k_i}$ for all $i$, then $\overline{\v}$ must be in cycle
$\left[L^{j_1}\u^1_{k_1}+L^{j_2}\u^2_{k_2}+\cdots+L^{j_s}\u^s_{k_s}\right]$. 

Thus, given any nonzero cycle $\Gamma_1$ in $\Omega(f(x))$ we can determine any of its state $\v$, 
find the conjugate state $\overline{\v}$, and the cycle $\Gamma_2$ that $\overline{\v}$ is a state of. If so desired, all conjugate pairs shared by any adjacent cycles can be determined explicitly. 
Finally, the steps detailed in~\cite[Sections IV and VI]{Chang2019} yield actual de Bruijn sequences.

Using $f(x)$ in (\ref{eq:prod}) may become crucial when substantially more than $\Lambda_n$ de Bruijn sequences of order a Mersenne exponent $n$ need to be produced. The simplest choice is to use $s=2$ with $f_1(x)$ an irreducible polynomial of a small degree, \eg, $1+x$ or $1+x+x^2$, and $f_2(x)$ any irreducible non-primitive polynomial of degree $n-1$ or $n-2$, respectively. If even more de Bruijn sequences are required, one should use $s \geq 3$ and choose a small $n_i$ for $i \in \bbra{1,s-1}$ since computing the Zech logarithm table relative to a small $n_i$ is easy.

\section{Zech's Logarithms and Cross-Join Pairs}\label{sec:cross}

Mykkeltveit and Szmidt showed in~\cite{MS15} that, starting from any de Bruijn sequence of a given order $n$, one can construct all de Bruijn sequences of that order by repeated applications of the cross-join pairing method. We now use Zech's logarithms to find the cross-join pairs of states. This enables us to construct the feedback functions of \emph{NLFSRs} of maximum period. The Fryers Formula for an $m$-sequence is a new analytic tool in the theory. We use the coefficients in the formula as an algorithmic tool here.

\subsection{Basic Notions and Results}
There exists a Boolean function $g$ on $n-1$ variables such that any non-singular feedback function $h$ satisfies $h(x_{0}, x_{1}, \ldots, x_{n-1}) = 
x_{0} + g(x_{1},\ldots,x_{n-1})$~\cite{Golomb81}. A {\it modified de Bruijn sequence} $\bS'$ of order $n$ is a sequence of length $2^{n}-1 $ obtained from a de Bruijn sequence $\bS$ of order $n$ by removing one zero from $\bS$'s tuple of $n$ consecutive zeros. 

\begin{theorem}(\cite{Golomb81}) 
Let $\bS = (s_0,s_1,\ldots,s_{2^n-1})$ be a de Bruijn sequence of order $n$. Then there exists a Boolean function $g(x_1,\ldots,x_{n-1})$ such that 
\[s_{t+n} = s_t + g(s_{t+1},\ldots,s_{t+n-1}) \mbox{ for } t \in \bbra{0,2^{n}-n-1}.\]
\end{theorem}

Let $(\bm{\alpha}=(a_0,\bm{A}), \overline{\bm{\alpha}}=(a_0+1,\bm{A}))$ and $(\bm{\beta}=(b_0,\bm{B}), \overline{\bm{\beta}}=(b_0+1,\bm{B}))$ be two conjugate pairs from a feedback function $h$, given $\bm{A}=(a_1,a_2,\ldots,a_{n-1}) \neq \bm{B}=(b_1,b_2,\ldots,b_{n-1})$. 
Then $(\bm{A},\bm{B})$ is the corresponding {\it cross-join pair} generated by the \emph{FSR} if the states $\bm{\alpha}, \bm{\beta}, \overline{\bm{\alpha}}, \overline{\bm{\beta}}$ occur in exactly that order. Let $\bS$ be a de Bruijn sequence of order $n$ generated by the feedback function $h$. Let $(\bm{A},\bm{B})$ be the cross-join pair. Then the feedback function
\begin{equation}\label{eq:ff_NLFSR} 
\widetilde{h}(x_0,\ldots,x_{n-1}) = h(x_0, x_1, \ldots , x_{n-1}) +  \prod^{n-1}_{i=1} (x_i + a_i + 1) + \prod^{n-1}_{i=1} (x_i + b_i + 1) 
\end{equation}
generates a new de Bruijn sequence $\bR$. The modified sequences $\bS'$ and $\bR'$ are similarly connected.

\begin{theorem}(\cite[Theorem 3.1]{MS15})
Let $\bS$ and $\bR$ be distinct de Bruijn sequences of order $n$. Then $\bS$  can be obtained from  $\bR$ by repeatedly applying the cross-join method.
\end{theorem}

\begin{example}
Consider the de Bruijn sequences of order $5$
\begin{align*}	
\bS &=(0000~0110~\underline{1100}~0100~1\overline{011~1}\underline{110~\overline{0}}\overline{111}~0101) \mbox{ and }\\
\bR &=(0000~0110~1110~1010~0101~1001~1111~0001).
\end{align*}	
Applying the method on $\bS$ using the pair $(\bm{A}=1100,\bm{B}=0111)$ results in an intermediate de Bruijn sequence $\mathbf{T}:=(0000~0110~110\overrightarrow{0~1}1\overrightarrow{1~1}10\overrightarrow{00}~1001~011\overrightarrow{1~0}101)$ where a right arrow indicates a crossover jump. We obtain $\bR$ by using the pair $(\bm{A}=1011,\bm{B}=0100)$ on $\mathbf{T}$. The respective feedback functions
\begin{align*}
h_{\bS} & = x_0 + x_1~ ( x_2 \cdot x_3 \cdot x_4 + x_2 \cdot x_4 + x_3 \cdot x_4 + x_3 + 1) 
+ x_2 ~(x_3 \cdot x_4 + x_3) + x_3 + 1, \\
h_{\mathbf{T}} & = x_0 + x_1~ ( x_2 \cdot x_3 \cdot x_4 + x_2 \cdot x_3 + x_2 + x_3 \cdot x_4 + x_3 + 1) + x_2 \cdot x_3 + x_3 + 1, \\
h_{\bR} & = x_0 + x_1~ ( x_2 \cdot x_3 \cdot x_4 + x_2 \cdot x_4 + x_3 + 1) +  
x_2 ~(x_3 \cdot x_4 + x_4+1) + x_3 + 1,
\end{align*}
satisfy Equation (\ref{eq:ff_NLFSR}). \qed
\end{example}

\subsection{Zech's Logarithms as Positional Markings}

The keen reader must have recognized by now that only a simple adaptation is required to determine cross-join pairs by using Zech's logarithms. Instead of identifying conjugate pairs whose respective components belong to distinct cycles, we use Zech's logarithms to identify states $\bm{\alpha}, \bm{\beta}, \overline{\bm{\alpha}}, \overline{\bm{\beta}}$ that occur in that exact order in a given (modified) de Bruijn sequence $\bS$. We retain the standard representation of states based on the $\varphi$ map with $1,\0^{n-1}$ as the initial state of any $m$-sequence of order $n$.

We demonstrate how to use Zech logarithms to recursively identify cross-join pairs in an $m$-sequence to generate numerous feedback functions of the NLFSRs that produce modified de Bruijn sequences. Remark~\ref{rem:3} comes in handy here. 

Let $\alpha$ be a root of the primitive trinomial $p(x) = x^{31} + x^3 +1$. Using the $\double$ map, one gets $\tau(3) = 31$ and $\tau(6) = 62$. We use the notation $ c = (3,6,31,62)$ to refer to two pairs of states 
$(\varphi(\alpha^{3}), \varphi(1+ \alpha^{3}=\alpha^{31}))$ and $(\varphi(\alpha^{6}), \varphi(1 + \alpha^{6} = \alpha^{62}))$. 
Hence, the states of the LFSR, with $p(x)$ as the characteristic polynomial, at the crossover jumps are governed by $\bm{\alpha}= (0,\0^{27},1,0,0)$ and $\bm{\beta}= (0,\0^{24},1,\0^5)$. Thus, $a_{28}=1$ and $b_{25}=1$ while $a_j=b_k=0$ for $1 \leq j \neq 28 \leq 30$ and $1 \leq k \neq 25 \leq 30$. The feedback function of the constructed NLFSR of period $2^{31}-1$ is, therefore, 
\[
\widetilde{h} = x_0 + x_3 +  \prod^{30}_{j=1} (x_j + a_j + 1) + \prod^{30}_{k=1} (x_k + b_k + 1)
\]
with algebraic degree $29$.

Now, consider the primitive trinomial $p(x) = x^{127} + x + 1$. 
From $\tau(1) = 127$ and $\tau(2) = 254$ one obtains the following set of mutually disjoint cross-join pairs  
\[
\left\lbrace c_i = \left(2^{8i}, 2^{1+8i}, 127 \cdot 2^{8i}, 127 \cdot 2^{1+8i}\right) \mbox{ for }
i \in \bbra{0,15} \right\rbrace.
\]
From this family we construct $ 2^{16}-1$ NLFSRs. Each of them generates a modified de Bruijn sequence of period $2^{127}-1$ whose feedback function has algebraic degree $125$.

As in the cycle joining method, we give a basic software implementation in~\cite{EFZech} with some randomness added in the procedure. The routine executes the following steps for any order $n$.
\begin{enumerate}
	\item Choose a primitive polynomial $p(x)$ of degree $n$ and generate the $m$-sequence $\bm$ from the initial state $1,\0^{n-1}$.
	\item Build the Zech's logarithm table based on $p(x)$. Note that in this simple implementation we use the approach suggested in Remark~\ref{rem:3}. For larger values of $n$ it is more efficient to obtain the logarithms from MAGMA or its alternatives.
	\item Pick random distinct integers $a$ and $b$ satisfying $a < b < \tau(a) < \tau(b)$.
	\item Construct the matrix $A_{p}$ from Equation (\ref{comx}) and use the square-and-multiply technique based on the binary representations of $a$ and $b$ to quickly find
	\[
	\bm{\alpha} = (1,0,0,\ldots,0) ~ A_{p}^a \mbox{ and } 
	\bm{\beta} = (1,0,0,\ldots,0) ~ A_{p}^b.
	\]
	\item The resulting feedback function of the modified de Bruijn sequence follows from applying Equation (\ref{eq:ff_NLFSR}), using $\bm{\alpha}$ and $\bm{\beta}$, with $h$ derived from $p(x)$.
\end{enumerate}

Here is a small instance to illustrate the process.
\begin{example}
On input $n=5$, the routine selected $p(x)=x^5 + x^2 + 1$. The chosen pair was $(a,b)=(3, 17)$ with 
$(\tau(a),\tau(b))=(22,25)$. Hence, $\bm{\alpha}=(0, 1, 0, 1, 1)$ and $\bm{\beta}=(0, 1, 1, 0, 1)$. The feedback function of the resulting modified de Bruijn sequence 
$(10000~10010~11101~10011~11100~01101~0)$ is $\widetilde{h} = x_0 + x_1 \cdot x_2 \cdot x_4 + x_1 \cdot x_3 \cdot x_4 + x_2$. \qed
\end{example}

\subsection{Fryers' Formula for $m$-Sequences}

Let $S(n) $ be the set of functions $f: \F_{2}^{n-1} \rightarrow \F_{2}$ such that $x_0 + f(x_1,\ldots,x_{n-1})$ generate de Bruijn sequences of order $n$. 
For a function $f : \F_{2}^{n-1} \rightarrow \F_{2}$, we define the set 
\[
S(f;k) \triangleq \{g \in S(n): \mbox{the weight of the function } f + g \mbox{ is } k \}.
\]
The number of inputs for which the functions 
$ f $ and  $ g $ are different equals $ k $. Let $ N(f;k)$ denote the cardinality of $S(f;k)$ and $G(f;y) \triangleq \sum_{k} N(f;k)y^k$. 
Let a linear function $ \ell : \F_{2}^{n-1} \rightarrow \F_{2} $ be such that $x_0 + \ell(x_1,\ldots,x_{n-1})$ generates an 
$ m $-sequence $\m$. In~\cite{CRV17}, Coppersmith, Rhoades, and Vanderkam proved the formula 
\begin{equation}\label{eq:formula}
G(\ell;y) = \frac{1}{2^n} \left((1 + y)^{2^{n-1}} - (1 - y)^{2^{n-1}} \right)
\end{equation}
which they attributed to M.~J.~Fryers. The coefficients  $ N(\ell;k) $ can be calculated by expanding the powers in Equation (\ref{eq:formula}). 

The value $ N(\ell;1) = 1$ implies that from $\m$ we obtain a unique de Bruijn sequence and joining the cycle of the zero state corresponds to one change in the truth table of the function $\ell$. In general, $ N(\ell;k) = 0 $ for all even $k$ since an even number of changes in the truth table of $\ell$ always lead to disjoint cycles. 

We discover an interesting combinatorial view on the non-vanishing ($>0$) coefficients of the polynomial $G(\ell;y)$. Recall that $\ell$ generates an $m$-sequence of period $2^n-1$. Sequence A281123 in OEIS~\cite{OEIS28} gives the formula for the positive coefficients of the polynomial
\[
G(\ell;y)=q(n-1,y) = \frac{(1+y)^{2^{n-1}} - (1-y)^{2^{n-1}}}{2^{n}}. 
\]
Hence, for odd $1 \leq k \leq 2^{n-1}-1$, the formula for $N(\ell;k)$ is  
\begin{equation}\label{eq:combi}
N(\ell;k) = \frac{1}{2^{n-1}} \, \binom{2^{n-1}}{k} \mbox{ for } n \geq 2.
\end{equation}
The number of cross-join pairs for an $m$-sequence, given by the Helleseth and Kl{\o}ve \cite{HK91} formula
\[
N(\ell;3) = \frac{(2^{n-1} - 1)(2^{n-1} - 2)}{3!},
\] 
follows from~(\ref{eq:combi}). The analogous formula for higher $k$  
\begin{equation}\label{eq:higher k}
N(\ell;k=2j+1\geq 5) = \frac{1}{k!} \prod_{i=1}^{k-1} (2^{n-1} - i)
\end{equation}
gives the number of de Bruijn sequences obtained after the $j$-th application of cross-join method. One starts from $\m$ and append $0$ to the longest string of zeroes to obtain a de Bruijn sequence $\bS$. One then finds all of its cross-join pairs and use them to construct new de Bruijn sequences. For each of the sequences repeat the cross-join method $j-1$ times.

Using (\ref{eq:combi}) we easily obtain the number of de Bruijn sequences of order $ n$:
\begin{equation}
G(\ell;1) = \sum_{k=1}^{2^{n-1}-1} N(\ell;k) = \frac{1}{2^{n-1}} \underbrace{\sum_{k=1}^{2^{n-1}-1} \binom{2^{n-1}}{k}}_{:=\Upsilon}=2^{2^{n-1} - n},
\end{equation}
since $\Upsilon=2^{2^{n-1}-1}$ is the sum of the odd entries in row $2^{n-1}$ of the Pascal Triangle. Note that the positive coefficients of $G(\ell;y)$ are symmetric. Let us examine some small orders. 

For $ n = 4$, one gets $\sum_{k=1}^{7} N(\ell;k) = 1 + 7 + 7 + 1 = 2^4$. Hence, from a de Bruijn sequence $\bS$, modified from $\m$, there are $7$ de Bruijn sequences each that can be constructed by using $1$, respectively $2$, application(s) of the cross-join method. There is a unique de Bruijn sequence that requires $3$ applications. For $ n = 5$, we know the number of de Bruijn sequences of distance $j=2k+1$ cross-join pairs from a fixed $m$-sequence of order $5$ from the formula 
\[
\sum_{k=1 \, ; \, k \mbox{ odd}}^{15} N(\ell;k) = 1 + 35 + 273 + 715 + 715 + 273 + 35 + 1 = 2^{11}.
\]

The Fryers formula for $n \geq 6$ can be easily generated by using Equation (\ref{eq:combi}). It tells us when to stop searching for more $j$ cross-join pairs and suggests the following naive approach to generating all de Bruijn sequences of order $n$. 
\begin{enumerate}
	\item Choose an $ m $-sequence $\m$ of period $2^n-1$ and the corresponding linear recursion $\ell$. Find all cross-join pairs in $\m$ and use Equation (\ref{eq:ff_NLFSR}) to construct the corresponding feedback fuctions for the resulting $ N(\ell;3)$ modified de Bruijn sequences.
	\item Find all cross-join pairs for each of the just constructed modified de Bruijn sequences and obtain feedback functions from those cross-join pairs. It may happen that some of the NLFSRs produced by different cross-join pairs are identical. A sieving mechanism is put in place in the second application to retain only distinct feedback functions. There are $N(\ell;5)$ of them.
	\item Repeat the application of the method to the end where only one new NLFSR appears. This ensures that the feedback functions of all modified de Bruijn sequences of order $n$ have been explicitly determined.
\end{enumerate}

This naive approach soon becomes impractical. Depending on the specific deployment needs, one can certainly throw in some random choices and various measures of stoppage to efficiently produce the feedback functions of only some de Bruijn sequences. The route via Zech's logarithms allows us to determine the algebraic degree of the resulting feedback function.

\section{Conclusions}\label{sec:con}

This work approaches the constructions of binary de Bruijn sequences for all orders $n$ via Zech's logarithms and demonstrates their practical feasibility. Many such sequences of large orders are produced by both the cycle joining method and the cross-join pairing within reasonable time and memory expenditures. As computational tools we establish further properties of Zech's logarithms as well as a rapid method to generate irreducible polynomials of degree $m \mid n$ from a primitive polynomial of degree $n$. We have also shown how to generate de Bruijn sequences by using LFSRs whose characteristic polynomial are products of distinct irreducible polynomials.

For a large $n$, storing a de Bruijn sequence of order $n$ is not feasible. Our proof-of-concept implementation determines the feedback functions of the resulting NLFSRs explicitly. This makes the approach highly adaptive to application-specific requirements. In the cycle joining framework one may opt to output subsequent states up to some length, \eg, measured in file size. The initial state and the spanning tree(s) can be randomly chosen. We also demonstrate how to use the Zech logarithms to quickly identify the cross-join pairs in any $m$-sequence, especially the one whose characteristic polynomial is a trinomial or a pentanomial. This enables us to determine the feedback functions of the NLFSRs that generate the resulting modified de Bruijn sequences. We present some consequences of the Fryers formula and use them as computational tools.

Our work also points to many open directions in the studies of NLFSRs that generate de Bruijn sequences or their modified sequences. A comprehensive analysis on their algebraic structure and complexity remains a worthy challenge to undertake.

\begin{acknowledgements}
M.~F.~Ezerman gratefully acknowledges the hospitality of the 
School of Mathematics and Statistics of Zhengzhou University, 
in particular Prof. Yongcheng Zhao, during several visits. 
The work of Z.~Chang is supported by the National Natural Science 
Foundation of China Grant 61772476 and the Key Scientific Research 
Projects of Colleges and Universities in Henan Province Grant 18A110029. 
Research Grants TL-9014101684-01 and MOE2013-T2-1-041 support the research carried out by M.~F.~Ezerman, S.~Ling, and H.~Wang. Singapore Ministry of Education Grant M4011381 provides a partial support for A.~A. Fahreza. J.~Szmidt's work is supported by the Polish Ministry of Science and Higher Education Decision No. 6663/E-294/S/2018.
\end{acknowledgements}



\end{document}